\documentclass[preprint,3p,times]{elsarticle}
\RequirePackage{multirow,booktabs,subfigure,color,array,hhline,makecell}
\usepackage{amssymb}
\usepackage{amsmath}
\usepackage{graphicx}
\usepackage{subfigure}
\usepackage{amsthm}
\usepackage{mathrsfs}
\usepackage{indentfirst}
\usepackage[colorlinks,citecolor=blue,urlcolor=blue]{hyperref}
\allowdisplaybreaks
\usepackage[table]{xcolor}
\usepackage{tikz-network}
\usepackage{pgf}
\usepackage{tikz}
\usepackage{subfigure}
\usepackage{matlab-prettifier}
\usetikzlibrary{arrows, decorations.pathmorphing, backgrounds, positioning, fit, petri, automata}
\usetikzlibrary{shadows,arrows,positioning}
\RequirePackage{algorithm,algpseudocode}
\allowdisplaybreaks
\usepackage{changes}
\usepackage{caption}
\usepackage{bbm}

\newtheorem{thm}{Theorem}
\newtheorem{defin}{Definition}
\newtheorem{lem}{Lemma}
\newtheorem{assum}{Assumption}
\newtheorem{rem}{Remark}

\newtheorem{con}{Condition}

\allowdisplaybreaks[4]
\AtBeginDocument{%
	\providecommand\BibTeX{{%
			\normalfont B\kern-0.5em{\scshape i\kern-0.25em b}\kern-0.8em\TeX}}}
\journal{Chaos, Solitons and Fractals}
\begin{document}
\captionsetup[figure]{labelfont={bf},labelformat={default},labelsep=period,name={Fig.}}
\begin{frontmatter}
\title{Discovering overlapping communities in multi-layer directed networks}
\author[label1]{Huan Qing\corref{cor1}}
\ead{qinghuan@u.nus.edu;qinghuan@cqut.edu.cn}
\cortext[cor1]{Corresponding author.}
\address[label1]{School of Economics and Finance, Chongqing University of Technology, Chongqing, 400054, China}
\begin{abstract}
Community detection in multi-layer undirected networks has attracted considerable attention in recent years. However, multi-layer directed networks are common in the real world, and existing community detection methods often either ignore the asymmetric structure in multi-layer directed networks or assume that every node solely belongs to a single community, significantly limiting their applicability to overlapping multi-layer directed networks, where nodes can belong to multiple communities simultaneously. To fill this gap, this article explores the challenging problem of detecting overlapping communities in multi-layer directed networks. Our goal is to understand the underlying asymmetric overlapping community structure by analyzing the mixed memberships of nodes. We introduce a novel multi-layer mixed membership stochastic co-block model (multi-layer MM-ScBM) to model overlapping multi-layer directed networks. We develop a spectral procedure to estimate nodes' memberships in both sending and receiving patterns. Our method uses a successive projection algorithm on a few leading eigenvectors of two debiased aggregation matrices. To our knowledge, this is the first work to detect asymmetric overlapping communities in multi-layer directed networks. We demonstrate the consistent estimation properties of our method by providing per-node error rates under the multi-layer MM-ScBM framework. Our theoretical analysis reveals that increasing the overall sparsity, the number of nodes, or the number of layers can improve the accuracy of overlapping community detection. Extensive numerical experiments validate these theoretical findings. We also apply our method to one real-world multi-layer directed network, gaining insightful results.
\end{abstract}
\end{frontmatter}
\section{Introduction}\label{sec1}
In recent years, multi-layer networks, also known as multiplex networks, have garnered considerable attention for their ability to model complex relationships and interactions within complex systems. A multi-layer network comprises numerous interconnected networks, each layer representing a distinct type of relationship or interaction for a common set of nodes. For example, in social network analysis, nodes may represent individuals, while distinct layers might reflect friendships, familial ties, or work associations. If the relationship between any two nodes is symmetric and undirected, these multi-layer networks are classified as multi-layer undirected networks. Conversely, when the relationship is asymmetric and directed, they are known as multi-layer directed networks \citep{su2024spectral}.

Community detection is a powerful tool for learning the latent community structure within networks. Its core objective is to extract community information of nodes solely based on the network's topology. For a comprehensive understanding of this technique, refer to \citep{fortunato2010community,fortunato2016community,javed2018community}, which provide detailed introductions to the field of community detection. Notably, in real-world networks, a node can often be affiliated with multiple communities \citep{MMSB,xie2013overlapping,OCCAM,mao2021estimating}. In the context of directed networks, where interactions between nodes are asymmetric and directional, each node exhibits two distinct community patterns: a sending pattern and a receiving pattern \citep{rohe2016co, su2024spectral}. The overlapping nature of nodes, their asymmetric relationships, and the multi-layer characteristics of multi-layer directed networks further complicate community detection. We focus on the challenging problem of discovering asymmetric overlapping communities in multi-layer directed networks in this article.
\begin{table}[h!]
\centering
\caption{A brief summary of community detection methods in multi-layer networks. Abbreviations: MLE (Maximum Likelihood Estimation), LSE (Least Squares Estimation), and NMF (Nonnegative Matrix Factorization).}
\label{LiteratureReview}
\resizebox{\columnwidth}{!}{
\begin{tabular}{ccccc}
\hline
Methods considered&Methods type&Directed networks?&Overlapping communities?&Theoretical guarantees?\\
\hline
Spectral Mean and Profile MLE \cite{han2015consistent}&Spectral-based and MLE-based&No&No&Yes\\
S2-jNMF \cite{ma2018community}&NMF-based&No&No&No\\
CSNMF and CSNMTF \cite{gligorijevic2018non}&NMF-based&No&No&No\\
SCDSBM \cite{pensky2019spectral}&Spectral-based&No&No&Yes\\
IterModMax \cite{pamfil2019relating}&Modularity-based&No&No&No\\
MLCD \cite{shahmoradi2019multilayer}&Multi-objective-optimization-based&No&Yes&No\\
Multimodbp \cite{weir2020multilayer}&Modularity-based&No&No&No\\
LSE \cite{lei2020consistent}&LSE-based &No&No&Yes\\
OLMF, Mean adj, Coreg spec, and SpecK \cite{paul2020spectral}&Matrix-factorization-based and spectral-based&No&No&Yes\\
Cr-ENMF \cite{ma2020co}&NMF-based&No&No&No\\
UASE \cite{gallagher2021spectral}&Spectral-based&No&No&Yes\\
TWIST \cite{jing2021community}&Tensor-based&No&No&Yes\\
MASE \cite{arroyo2021inference}&Spectral-based&No&No&Yes\\
MjNMF \cite{ma2021identification}&NMF-based&No&No&No\\
NE2NMF \cite{li2021detecting}&NMF-based&No&No&No\\
MA, MVM, and MVP \cite{venturini2022variance}&Modularity-based&No&No&No\\
ALMA \cite{fan2022alma}&Tensor-based&No&No&Yes\\
SSJSNMF \cite{lv2022community}&NMF-based&No&No&No\\
SoS-Debias \cite{lei2023bias}&Spectral-based&No&No&Yes\\
CAMSBM \cite{xu2023covariate}&Tensor-based&No&No&Yes\\
ISC \cite{huang2023spectral}&Spectral-based&No&No&Yes\\
HDP-LPCM \cite{daniel2023bayesian}&Bayesian-based&No&No&No\\
McNMF \cite{jia2023clustering}&NMF-based&No&No&No\\
CDMA \cite{li2023multiplex}&Motif-based&No&No&No\\
MOEA-CPI \cite{gao2023multilayer}&Multi-objective-evolutionary-based&No&No&No\\
CDMMHN \cite{liu2024motif}&Motif-modularity-based&No&No&No\\
WPTDD \cite{peng2024weighted}&Tensor-based&No&No&No\\
HSBM \cite{amini2024hierarchical}&Bayesian-based&No&No&No\\
DIMPLE-GDPG \cite{pensky2024clustering}&Tensor-based&No&No&Yes\\
MX-ONMTF \cite{ortiz2024community}&NMF-based&No&No&No\\
SPATMME \cite{agterberg2024estimating}&Tensor-based&No&Yes&Yes\\
MCLCD \cite{cao2024multi}&NMF-based&No&No&No\\
NMFTD \cite{tang2025detecting}&NMF$\&$Tensor-based&No&Yes&Yes\\
NSoA and NDSoSA \cite{qingMLDCSBM}&Spectral-based&No&No&Yes\\
DSoG \cite{su2024spectral}&Spectral-based&Yes&No&Yes\\
CSPDSoS (ours)&Spectral-based&Yes&Yes&Yes\\
\hline
\end{tabular}}
\end{table}

A widely studied variant of this problem focuses on the scenario where the network is undirected, and each node exclusively belongs to a single community. For example, the multi-layer stochastic block model (multi-layer SBM), as explored in various studies \citep{han2015consistent,paul2020spectral,lei2020consistent,lei2023bias}, generates each layer of a multi-layer undirected network from the stochastic block model (SBM) \citep{holland1983stochastic}, a popular statistical model for single-layer undirected networks in which communities are mutually exclusive. \citep{han2015consistent} demonstrated the estimation consistency of a spectral approach designed based on the sum of adjacency matrices, specifically under conditions where the number of layers increases while the number of nodes remains fixed, within the framework of the multi-layer SBM. \citep{paul2020spectral} conducted a study on the theoretical guarantees of spectral and matrix factorization methods within the context of the multi-layer SBM. \citep{lei2020consistent} introduced a least squares estimation method for communities and demonstrated its estimation consistency within the multi-layer SBM. Notably, the debiased spectral clustering algorithm introduced in \citep{lei2023bias}, which offers theoretical guarantees, significantly outperforms the spectral method using the sum of adjacency matrices, the Linked Matrix Factorization method, and the Co-regularized Spectral Clustering method studied in \citep{paul2020spectral}. Also see \citep{jing2021community,arroyo2021inference,fan2022alma,xu2023covariate} for other recent works in multi-layer undirected networks in which communities are disjoint.

Unluckily, applying all aforementioned algorithms to discover communities in multi-layer directed networks is not feasible. To address this limitation, \citep{su2024spectral} proposed a spectral co-clustering algorithm that leverages the debiased approach introduced in \citep{lei2023bias} to identify disjoint communities under the multi-layer stochastic co-block model (multi-layer ScBM). This model extends the stochastic co-block model (ScBM) presented in \citep{rohe2016co} from single-layer directed networks to multi-layer directed networks. While the works in \citep{su2024spectral} address the asymmetric and directional characteristics in multi-layer directed networks, they do not extend to the general problem addressed in this article, where a node may simultaneously belong to more than one community. Recently, \citep{agterberg2024estimating,tang2025detecting} introduced tensor-based methods with theoretical guarantees for estimating mixed memberships in overlapping multi-layer undirected networks. Furthermore, various other approaches for community detection in multi-layer networks exist, such as NMF-based and modularity-based techniques, as summarized in Table \ref{LiteratureReview}. Spectral-based methods, such as those proposed in \citep{han2015consistent,lei2023bias,su2024spectral}, which rely on the eigendecomposition or singular value decomposition of certain matrices, are also included in this table. However, as Table \ref{LiteratureReview} clearly indicates, all previous methods either focus on non-overlapping communities or are tailored for multi-layer undirected networks. To the best of our knowledge, there is currently no efficient method available for estimating nodes' mixed memberships in overlapping multi-layer directed networks, where the relationships between nodes are directional and each node can belong to multiple communities with varying weights. This article aims to fill this gap. The key contributions of this article are as follows:

(1) We propose a flexible and interpretable statistical model, the multi-layer mixed membership stochastic co-block model (multi-layer MM-ScBM). Our model enables nodes within multi-layer directed networks to be affiliated with multiple communities.

(2) We propose a spectral method by running a vertex hunting algorithm \citep{gillis2013fast} on a few leading eigenvectors of two debiased aggregation matrices to estimate nodes' mixed memberships for both the sending and receiving patterns under the multi-layer MM-ScBM framework. To our knowledge, our work is the first to discover overlapping communities in multi-layer directed networks.

(3) By deriving node-wise error bounds, we demonstrate that our method ensures consistent estimation. Our theoretical results underscore the advantages of augmenting overall sparsity, the number of nodes, or the number of layers on the problem of overlapping community detection in multi-layer directed networks. These theoretical results are verified by synthetic data, and we further illustrate the practical applications of our method by considering one real-world multi-layer directed network with encouraging results.

The structure of the remaining sections of this article is outlined in the following manner. Section \ref{RelatedWorks} covers related works. Section \ref{sec3} introduces the model. Section \ref{sec4} presents the spectral method. Section \ref{sec5} presents the consistency results. Section \ref{sec6experiments} includes both simulations and practical applications. Finally, Section \ref{sec7} concludes the article. Technical proofs are in \ref{SecProofs}.

\emph{Notation:} For any positive integer $q$, let $[q]$ be the set $\{1,2,\ldots,q\}$, $I_{q\times q}$ be the $q\times q$ identity matrix, and $\|x\|_{q}$ be the $l_{q}$ norm of vector $x$. For any matrix $M$, we use $M'$ to denote its transpose, $\kappa(M)$ for its condition number, $\|M\|_{F}$ for the Frobenius norm, $\mathrm{rank}(M)$ for its rank, $M(S,:)$ (and $M(:,S)$) for the submatrix comprising the rows (and columns) of $M$ indexed by set $S$, $\lambda_{k}(M)$ for the $k$-th largest eigenvalue (ordered in absolute value), $\|M\|_{\infty}$ for the maximum absolute row sum, and $\|M\|_{2\rightarrow\infty}$ for the maximum row-wise $l_{2}$ norm. We further denote expectation with $\mathbb{E}[\cdot]$ and probability with $\mathbb{P}(\cdot)$. Additionally, the phrase ``leading $K$ eigenvectors" refers to the eigenvectors respective to the top $K$ eigenvalues ordered in absolute value. Lastly, $e_{i}$ represents the standard basis vector that has the $i$-th element equal to 1, while all other elements are set to zero.
\section{Related work}\label{RelatedWorks}
The field of community detection in multi-layer networks has garnered significant attention, with various approaches being proposed in the literature to address the complexity of identifying latent community structures within such networks. In this section, we provide a structured discussion of relevant works, highlighting their connections and differences with respect to our work in this article.
\subsection{Spectral-based methods}
Recall that Table \ref{LiteratureReview} summarizes substantial existing works, including spectral-based, NMF-based, and modularity-based methods, for community detection in multi-layer networks. However, only a few works listed in this table are closely related to this study because this article focuses specifically on spectral-based methods due to their well-known efficiency, ease of implementation, and feasibility for theoretical analysis \citep{ng2001spectral, von2007tutorial,lei2015consistency}.

Within the spectral-based methods for multi-layer networks, our proposed method has a close relationship with the debiased spectral clustering methodologies presented in  \citep{lei2023bias,su2024spectral,qingMLDCSBM}. For non-overlapping community detection in multi-layer undirected networks, \citep{lei2023bias} initially proposes a debiased spectral clustering algorithm by applying the K-means algorithm to several leading eigenvectors of an aggregation matrix under the multi-layer SBM model. \citep{su2024spectral} extends the idea in \citep{lei2023bias} from multi-layer undirected networks modeled by multi-layer SBM to multi-layer directed networks modeled by multi-layer ScBM. \citep{qingMLDCSBM} extends works in \citep{lei2023bias} by considering nodes' degree variations in the real world, enhancing the debiased spectral clustering framework. The numerical findings in \citep{lei2023bias, su2024spectral,qingMLDCSBM} indicate that debiased spectral clustering achieves superior performance compared to traditional spectral clustering methods that rely on the summation of adjacency matrices or the squared summation of adjacency matrices. In particular, \cite{qingMLDCSBM} theoretically proved the superiority of the debiased spectral method over the spectral mean method \citep{han2015consistent}. However, works in \citep{lei2023bias, su2024spectral,qingMLDCSBM} are limited to non-overlapping communities, failing to capture the overlapping community structure prevalent in real-world networks.
\subsection{Overlapping community detection}
While the spectral-based methods proposed in \citep{lei2023bias, su2024spectral,qingMLDCSBM} have been successful in detecting non-overlapping communities in multi-layer networks, the task of identifying overlapping communities is more challenging. The mixed membership stochastic block model (MMSB) \citep{MMSB} allows each node to belong to multiple communities for single-layer undirected networks, and it is one of the most popular statistical models in network science due to its generality, flexibility, and interpretability. The multi-way blockmodel introduced in \cite{airoldi2013multi} extends MMSB from undirected networks to directed networks.

Spectral-based methods have also been developed for overlapping community detection in single-layer networks. For example, based on the observation that there exists an ideal simplex structure inherent in the leading eigenvector matrix of the expectation adjacency matrix, \citep{mao2021estimating} develops an efficient spectral algorithm to estimate overlapping community memberships and derives the node-wise error bounds under MMSB.   \citep{jin2024mixed} develops an algorithm called Mixed-SCORE with theoretical guarantees to fit a degree-corrected version of MMSB, where Mixed-SCORE is also designed based on an ideal simplex structure inherent in the row-wise ratio matrix of the leading eigenvector matrix from the expectation adjacency matrix. The Mixed-SCORE algorithm can be viewed as an extension of the SCORE algorithm developed in \cite{SCORE} from non-overlapping networks to overlapping networks. \citep{qing2023regularized} designs two regularized spectral clustering algorithms based on a regularized Laplacian matrix to estimate nodes' memberships under MMSB. For works on overlapping community detection in \citep{mao2021estimating,qing2023regularized,jin2024mixed}, vertex hunting algorithms developed in \citep{araujo2001successive,gillis2013fast,gillis2015semidefinite,jin2024mixed} are applied to hunt for pure nodes, as defined later.
\subsection{Connections to previous work}
Despite the progress made in overlapping community detection, existing methods are primarily designed for single-layer networks. To address the limitations of existing approaches, we first propose a novel statistical model multi-layer MM-ScBM to model overlapping communities in multi-layer directed networks. This model extends the MMSB and multi-way blockmodel from single-layer networks to multi-layer directed networks, capturing the asymmetric and directional relationships between nodes. We then develop a spectral-based method to estimate nodes' mixed memberships in both the sending and receiving patterns under the multi-layer MM-ScBM framework. Our method integrates the debiased idea in \citep{lei2023bias,su2024spectral,qingMLDCSBM} with the vertex hunting technique in \citep{mao2021estimating} to estimate nodes's mixed memberships under the proposed model effectively. Our work builds upon and extends the research conducted in \citep{lei2023bias}, \citep{su2024spectral}, and \citep{mao2021estimating}, transitioning from non-overlapping multi-layer undirected networks, non-overlapping multi-layer directed networks, and overlapping single-layer undirected networks to overlapping multi-layer directed networks, respectively.

To demonstrate the effectiveness of our proposed method, we theoretically reveal its consistent estimation properties. Our theoretical framework leverages the Bernstein inequality from \cite{tropp2012user} and Theorem 4.2 from \citep{cape2019the} to derive the row-wise eigenspace error for two debiased aggregation matrices. Subsequently, we employ theoretical frameworks from \citep{mao2021estimating} to establish row-wise error bounds. By doing so, we broaden the scope of the theoretical analysis in \citep{mao2021estimating}, extending it from single-layer undirected networks to multi-layer directed networks, and from adjacency matrices to debiased aggregation matrices.
\section{Model framework: multi-layer MM-ScBM}\label{sec3}
Recall that few previous works have focused on detecting the latent overlapping community structures in multi-layer directed networks, to fill this gap, we need to solve the following fundamental questions:

\begin{itemize}
  \item \textbf{Question 1:} How can we develop a statistical model to describe/capture/model multi-layer directed networks with latent overlapping community property such that nodes can have mixed community memberships?
  \item \textbf{Question 2:} Assuming we indeed propose a statistical model to address Question 1, how can we design an efficient algorithm to estimate the true mixed memberships of nodes in multi-layer directed networks generated from this model?
  \item \textbf{Question 3:} Assuming we indeed devise an algorithm to address Question 2, how can we guarantee that this algorithm enjoys consistent estimation properties under the proposed model? Here, consistent estimation properties indicate that the theoretical upper bound of the error rate of the algorithm decreases to zero as the number of nodes or network layers increases to infinity under this proposed model.
  \item \textbf{Question 4:} Assuming we theoretically demonstrate the consistent estimation of the proposed algorithm under the proposed model and derive the theoretical upper bound on its error rate, how can we numerically validate the effectiveness of the proposed algorithm?
\end{itemize}

The rest of this section details the statistical model to address Question 1, while the subsequent three sections present the algorithm, consistency results, and numerical experiments to address Questions 2, 3, and 4, respectively.

\begin{figure}
\centering
    \includegraphics[width=0.8\textwidth]{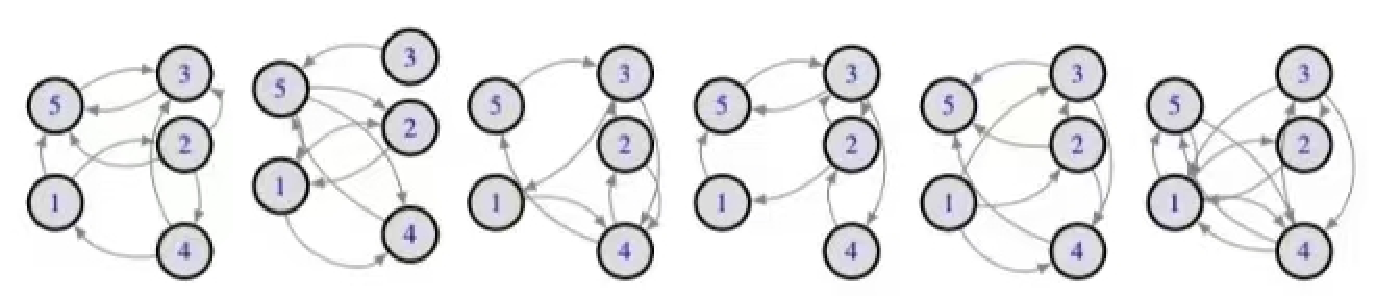}
\caption{A toy multi-layer directed network with 5 nodes and 6 layers.}
\label{biExample}
\end{figure}
This article assumes that the multi-layer directed network $\mathcal{N}$ comprises $L$ distinct layers, each sharing a common set of $n$ interconnected nodes. Let $A_{l} \in \{0,1\}^{n \times n}$ denote the adjacency matrix of the $l$-th layer directed network $\mathcal{N}_{l}$, where $A_{l}(i,j) = 1$ indicates the presence of a directed edge from node $i$ to node $j$ (i.e., $i \rightarrow j$) and $A_{l}(i,j) = 0$ otherwise, for all $i, j \in [n]$ and $l \in [L]$. Therefore, in $A_{l}$, its $i$-th row captures the outgoing edges from node $i$, while its $i$-th column details the incoming edges to node $i$ at the $l$-th layer for $i\in[n], l\in[L]$. Fig.~\ref{biExample} displays a toy example of a multi-layer directed network, showcasing how nodes establish distinct connections across various layers. For the sending pattern, we assume that all nodes are classified into $K$ distinguishable row communities:
\begin{align}\label{RowCommunities}
C_{r,1}, C_{r,2}, \ldots, C_{r,K}.
\end{align}

Let $\Pi_{r}$ be an $n\times K$ row membership matrix common to all layers. Specifically, $\Pi_{r}(i,k)$ represents the ``weight" indicating the affiliation of node $i$ to the $k$-th row community $C_{r,k}$ for $i\in[n], k\in[K]$.

Analogously, for the receiving pattern, we assume that all nodes are categorized into $K$ column communities:
\begin{align}\label{ColumnCommunities}
C_{c,1}, C_{c,2}, \ldots, C_{c,K}.
\end{align}

Furthermore, let $\Pi_{c}$ be an $n\times K$ column membership matrix common across all layers and $\Pi_{c}(i,k)$ represents the ``weight" indicating the affiliation of node $i$ to the $k$-th column community $C_{c,k}$ for $i\in[n], k\in[K]$.

Because $\Pi_{r}(i,:)$ is the row membership vector of node $i$, we have $\|\Pi_{r}(i,:)\|_{1}=1$ and $\Pi_{r}(i,k)\in[0,1]$ for $i\in[n], k\in[K]$. Similarly, we have $\|\Pi_{c}(i,:)\|_{1}=1$ and $\Pi_{c}(i,k)\in[0,1]$.

A node $i$ is considered ``pure" in the sending pattern if $\Pi_{r}(i,:)$ has exactly one entry equal to 1 and the remaining $(K-1)$ entries are 0. Conversely, it is deemed ``mixed" if this condition does not hold for $i\in[n]$. A similar definition holds for nodes in the receiving pattern. For the rest of this article, we assume that $K$ is known and
\begin{align}\label{PureNodes}
\text{each community must contain at least one pure node.}
\end{align}

Without loss of generality, we rearrange the nodes' order such that $\Pi_{r}(\mathcal{I}_{r},:) = I_{K\times K}$ and $\Pi_{c}(\mathcal{I}_{c},:) = I_{K\times K}$. Here, $\mathcal{I}_{r}$ represents the indices of $K$ pure nodes, each belonging to a distinct row community, while $\mathcal{I}_{c}$ denotes the indices of $K$ pure nodes, each from a unique column community.

For $l\in[L]$, let $B_{l}$ be a $K\times K$ heterogeneous block probability matrix such that
\begin{align}\label{Bl}
B_{l}(k,\tilde{k})\in[0,1]\qquad k\in[K],\tilde{k}\in[K].
\end{align}

For each $l$, $B_{l}$ can be asymmetric or symmetric and it captures the community-wise edge probabilities within this layer \citep{lei2023bias,su2024spectral}. In this article, $B_{l}$ can be different for different layers.

For each adjacency matrix $A_{l}$ for $l\in[L]$, the model considered in this article assumes that it is generated from a directed version of MMSB. In particular, we consider the following model.
\begin{defin}\label{MLMM-ScBM}
(multi-layer MM-ScBM) For $l\in[L]$, our multi-layer mixed membership stochastic co-block model (multi-layer MM-ScBM) generates the $l$-th adjacency matrix $A_{l}$ in the following way:
\begin{align}\label{OmegaL}
\Omega_{l}:=\rho\Pi_{r}B_{l}\Pi'_{c}~~~~~~~~~A_{l}(i,j)\sim\mathrm{Bernoulli}(\Omega_{l}(i,j))\qquad i\in[n], j\in[n],
\end{align}
where $\rho\in(0,1]$.
\end{defin}

In the definition of multi-layer MM-ScBM, given that the heterogeneous block probability matrix $B_{l}$ can be different for different layers, this guarantees that $A_{l}$ may be generated from the model with common membership matrices and possibly different edge probabilities. For $i\in[n], j\in[n], l\in[L]$, recall that $\Pi_{r}\in[0,1]^{n\times K}, \Pi_{c}\in[0,1]^{n\times K}, \|\Pi_{r}(i,:)\|_{1}=1, \|\Pi_{c}(i,:)\|_{1}=1$, and $B_{l}\in[0,1]^{K\times K}$, by basic algebra, we have $0\leq\Pi_{r}(i,:)B_{l}\Pi'_{c}(j,:)\leq1$. Then, under the multi-layer MM-ScBM model defined in Equation (\ref{OmegaL}), we have
\begin{align*}
&\mathbb{P}(A_{l}(i,j)=1)=\Omega_{l}(i,j)=\rho\Pi_{r}(i,:)B_{l}\Pi'_{c}(j,:)\leq\rho,\\
&\mathbb{P}(A_{l}(i,j)=0)=1-\Omega_{l}(i,j)=1-\rho\Pi_{r}(i,:)B_{l}\Pi'_{c}(j,:)\geq1-\rho,
\end{align*}
which implies that a reduction in the value of $\rho$ leads to a decrease in the likelihood of an edge forming between any two nodes, thereby controlling the overall sparsity of the multi-layer bipartite network. Consequently, we refer to $\rho$ as the sparsity parameter in this article. Fig.~\ref{PircB123Omega123} presents a toy example to facilitate a better understanding of Equation (\ref{OmegaL}). In this example, we illustrate the settings of $\Pi_{r}$, $\Pi_{c}$, $\{B_{l}\}^{L}_{l=1}$, and $\rho$. Using Equation (\ref{OmegaL}), we can immediately derive $\{\Omega_{l}\}^{L}_{l=1}$. From this figure, it is straightforward to observe that all elements of $\Omega_{l}$ lie within the range $[0,1]$ for $\l\in[L]$, given that $\Omega_{l}$'s elements represent probabilities. Given the set of $L$ adjacency matrices $\{A_{l}\}^{L}_{l=1}$, our goal is to estimate $\Pi_{r}$ and $\Pi_{c}$ in this article. Similar to the consistency analysis conducted in \citep{lei2015consistency,mao2021estimating,qing2023regularized,qing2024bipartite} for single-layer networks, we will permit $\rho$ to approach zero as $n$ and $L$ increase in our theoretical analysis.
\begin{figure}
\centering
\resizebox{\columnwidth}{!}{
    \includegraphics[width=0.02\textwidth]{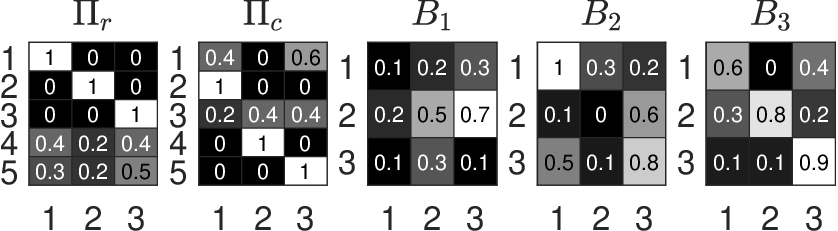}
}
\resizebox{\columnwidth}{!}{
    \includegraphics[width=2\textwidth]{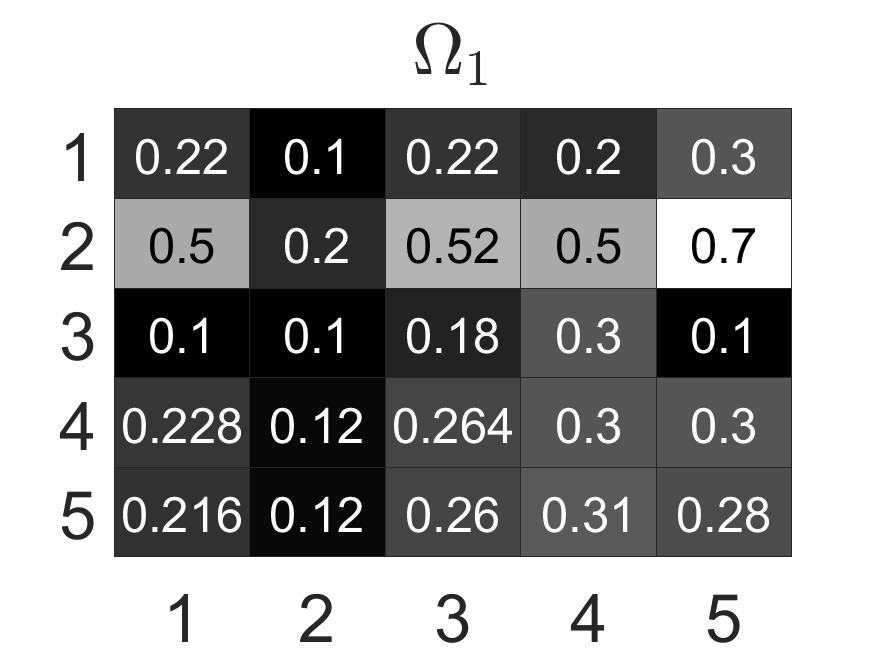}
    \includegraphics[width=2\textwidth]{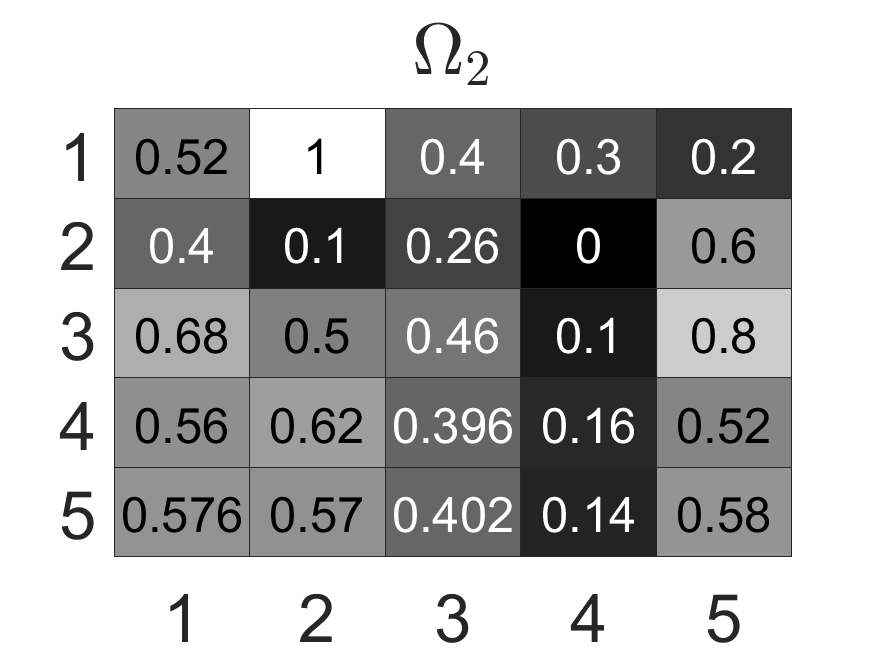}
    \includegraphics[width=2\textwidth]{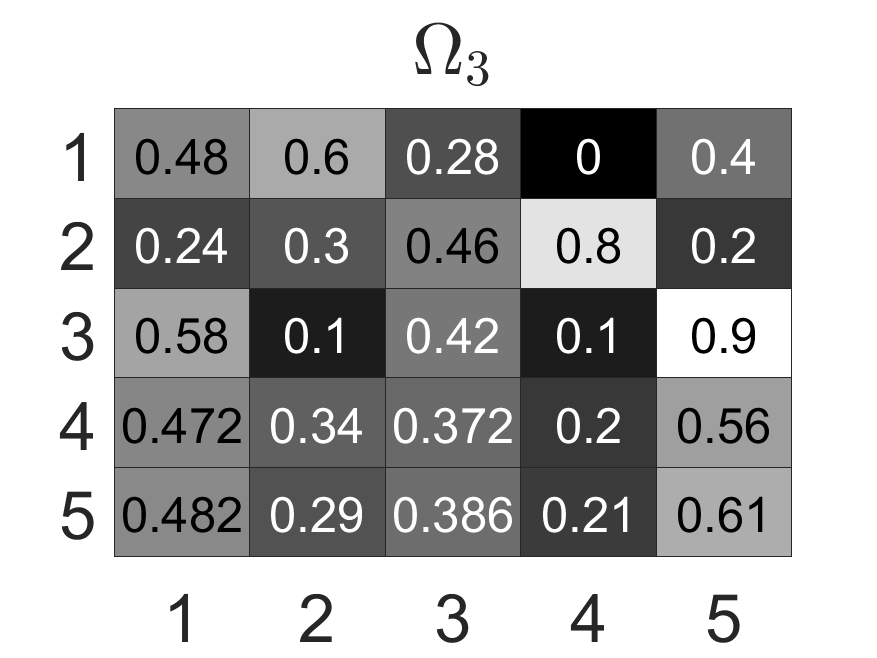}
}
\caption{A toy example for the settings of $\Pi_{r}$, $\Pi_{c}$, and $\{B_{l}\}^{L}_{l=1}$. The sparsity parameter $\rho$ is set as 1 here. In this example, we consider a 3-layer directed network with 5 nodes belonging to 3 row (and column) communities with different weights. The row and column membership matrices are $\Pi_{r}$ and $\Pi_{c}$, respectively. One can easily check that $\Omega_{l} = \rho \Pi_{r} B_{l} \Pi'_{c}$ for $l \in [L]$.}
\label{PircB123Omega123}
\end{figure}
\begin{rem}
Our model encompasses numerous previous models as specific cases. When all nodes are pure, it simplifies to the multi-layer ScBM discussed in \cite{su2024spectral}. If the network is undirected and all nodes are pure, it reduces to the multi-layer SBM studied in \cite{han2015consistent, paul2020spectral,lei2020consistent, lei2023bias}. Further, for $L=1$, it degenerates to the multi-way blockmodel presented in \cite{airoldi2013multi}. In the case of $L=1$ and an undirected network, it reduces to the MMSB model proposed in \cite{MMSB}. Moreover, when all nodes are pure and the network only has a single layer, it reduces to the ScBM model described in \cite{rohe2016co}. Finally, if the network is undirected, it simplifies to the SBM introduced in \cite{holland1983stochastic}.
\end{rem}
\section{Methodology: CSPDSoS}\label{sec4}
In this section, we introduce an efficient spectral method to estimate membership matrices $\Pi_{r}$ and $\Pi_{c}$ for adjacency matrices $\{A_{l}\}^{L}_{l=1}$ generated from the multi-layer MM-ScBM by Equation (\ref{OmegaL}). To explain the main idea behind our method, we provide an ideal method with knowing the $L$ population adjacency matrices $\{\Omega_{l}\}^{L}_{l=1}$ in Section \ref{sec3oracle} and modify it to the real case with observing the $L$ adjacency matrices $\{A_{l}\}^{L}_{l=1}$ in Section \ref{sec3real}.
\subsection{Oracle case}\label{sec3oracle}
For the oracle case when $\{\Omega_{l}\}^{L}_{l=1}$ are known, we define two aggregate matrices $\tilde{S}_{r}$ and $\tilde{S}_{c}$: $\tilde{S}_{r}=\sum_{l\in[L]}\Omega_{l}\Omega'_{l}$ and $\tilde{S}_{c}=\sum_{l\in[L]}\Omega'_{l}\Omega_{l}$. The following lemma is the starting point of our method.
\begin{lem}\label{EigOsum2}
Given that the rank of $\sum_{l\in[L]}B_{l}B'_{l}$ and $\sum_{l\in[L]}B'_{l}B_{l}$ is $K$, the following conclusions can be drawn:
\begin{itemize}
  \item Let $U_{r}$ represent the leading $K$ eigenvectors of $\tilde{S}_{r}$ such that $U'_{r}U_{r} = I_{K\times K}$. Then, it holds that $U_{r} = \Pi_{r}U_{r}(\mathcal{I}_{r},:)$.
  \item Similarly, let $U_{c}$ represent the leading $K$ eigenvectors of $\tilde{S}_{c}$ satisfying $U'_{c}U_{c} = I_{K\times K}$. It follows that $U_{c} = \Pi_{c}U_{c}(\mathcal{I}_{c},:)$.
\end{itemize}
\end{lem}

The form of $U_{r}=\Pi_{r}U_{r}(\mathcal{I}_{r},:)$ is known as ideal simplex, where the rows of $U_{r}$ locate in a simplex in $\mathbb{R}^{K}$ with vertices being the $K$ rows of $U_{r}(\mathcal{I}_{r},:)$. Such ideal simplex has been observed in single layer mixed membership community detection \citep{panov2018consistent,mao2021estimating,qing2023regularized,jin2024mixed,qing2024bipartite}, topic modeling \citep{klopp2023assigning,ke2024using}, and latent class analysis \citep{qing2023finding,chen2024spectral}. Since $U_{r}(\mathcal{I}_{r},:)$ is a $K\times K$ full rank matrix, we have $\Pi_{r}=U_{r}U^{-1}_{r}(\mathcal{I}_{r},:)$, which indicates that we can exactly recover the row membership matrix through $U_{r}U^{-1}_{r}(\mathcal{I}_{r},:)$ as long as the index set $\mathcal{I}_{r}$ is available. Benefited from the simplex structure $U_{r}=\Pi_{r}U_{r}(\mathcal{I}_{r},:)$, applying the \emph{Successive Projection Algorithm} (SPA) \citep{araujo2001successive,gillis2013fast,gillis2015semidefinite} to all rows of $U_{r}$ with $K$ clusters can exactly hunt for the $K$ vertices (i.e., the $K$ rows of $U_{r}(\mathcal{I}_{r},:)$) in the simplex. The details of SPA are provided in Algorithm \ref{alg:SP}. Upon execution on $U_{r}$, SPA accepts $U_{r}$ and $K$ as inputs (with $M=U_{r}$ and $r=K$ in Algorithm \ref{alg:SP}). SPA iteratively chooses $K$ rows from the eigenvector matrix $U_{r}$ that possess the highest Frobenius norms. Each selected row prompts SPA to eliminate its contribution within the matrix, thereby ensuring the orthogonality of subsequent selections to previously chosen rows. Consequently, a collection of $K$ pure nodes $\mathcal{I}$, is obtained. For deeper insights into SPA, please refer to \citep{mao2021estimating,jin2024mixed,chen2024spectral}.

\begin{algorithm}
\caption{SPA}
\label{alg:SP}
\begin{algorithmic}[1]
\Require Matrix $M\in\mathbb{R}^{n\times K}$ and integer $r\leq n$.
\Ensure Set of indices $\mathcal{I}\subseteq[n]$.
\State Initialize: $\mathcal{I}=\{\}$.
\State For $k\in[r]$ do:
\State ~~~~~~~$k_{*}=\mathrm{argmax}_{i\in[n]}\|M(i,:)\|_{F}$.
\State ~~~~~~$M\leftarrow M-M\frac{M'(k_{*},:)M(k_{*},:)}{\|M(k_{*},:)\|^{2}_{F}}$.
\State ~~~~~~$\mathcal{I}=\mathcal{I}\cup \{k_{*}\}$.
\State end for
\end{algorithmic}
\end{algorithm}

Similarly, $U_{c}=\Pi_{c}U_{c}(\mathcal{I}_{c},:)$ also forms an ideal simplex structure and applying the SPA to $U_{c}$ can exactly recover $\mathcal{I}_{c}$. The following four-stage algorithm called ideal \textbf{C}o-clustering by \textbf{S}equential \textbf{P}rojection on the \textbf{D}ebiased \textbf{S}um \textbf{o}f \textbf{S}quared matrices (ideal CSPDSoS) summarizes the above analysis. Input: $\{\Omega_{l}\}^{L}_{l=1}$ and $K$. Output: $\Pi_{r}$ and $\Pi_{c}$.
\begin{itemize}
  \item \emph{Aggregation step}. Set $\tilde{S}_{r}=\sum_{l\in[L]}\Omega_{l}\Omega'_{l}$ and $\tilde{S}_{c}=\sum_{l\in[L]}\Omega'_{l}\Omega_{l}$.
  \item \emph{Eigendecomposition step}. Obtain $U_{r}$ and $U_{c}$, the leading $K$ eigenvectors of $\tilde{S}_{r}$ and $\tilde{S}_{c}$, respectively.
  \item \emph{Vertex hunting step}. Run SPA to $U_{r}$'s (and $U_{c}$'s) rows with $K$ clusters to get $\mathcal{I}_{r}$ (and $\mathcal{I}_{c}$).
  \item \emph{Membership reconstruction step}. Set $\Pi_{r}=U_{r}U^{-1}_{r}(\mathcal{I}_{r},:)$ and $\Pi_{c}=U_{c}U^{-1}_{c}(\mathcal{I}_{c},:)$.
\end{itemize}
\subsection{Real case}\label{sec3real}
In practice, we observe the $L$ adjacency matrices $\{A_{l}\}^{L}_{l=1}$ instead of their expectations $\{\Omega_{l}\}^{L}_{l=1}$. Let $D^{r}_{l}$ and $D^{c}_{l}$ be two $n\times n$ diagonal matrices such that $D^{r}_{l}(i,i)=\sum_{j\in[n]}A_{l}(i,j)$ and $D^{c}_{l}(i,i)=\sum_{j\in[n]}A_{l}(j,i)$ for $i\in[n], l\in[L]$. Define $S_{r}$ and $S_{c}$ as $S_{r}=\sum_{l\in[L]}(A_{l}A'_{l}-D^{r}_{l})$ and $S_{c}=\sum_{l\in[L]}(A'_{l}A_{l}-D^{c}_{l})$. Drawing from the insightful analysis conducted in \citep{lei2023bias,su2024spectral}, we know that $S_{r}$ and $S_{c}$ are debiased estimations of $\tilde{S}_{r}$ and $\tilde{S}_{c}$, respectively. Let $\hat{U}_{r}$ and $\hat{U}_{c}$ be the leading $K$ eigenvectors of $S_{r}$ and $S_{c}$, respectively. We see that $\hat{U}_{r}$ and $\hat{U}_{c}$ can be viewed as good approximations of $U_{r}$ and $U_{c}$, respectively. Then applying SPA to all rows of $\hat{U}_{r}$ (and analogously to $\hat{U}_{c}$) with $K$ clusters should obtain accurate estimations of the index sets $\mathcal{I}_{r}$ (and $\mathcal{I}_{c}$). Algorithm \ref{alg:CSPDSoS} summarizes the above analysis.
\begin{algorithm}
\caption{\textbf{C}o-clustering by \textbf{S}equential \textbf{P}rojection on the \textbf{D}ebiased \textbf{S}um \textbf{o}f \textbf{S}quared matrices (CSPDSoS)}
\label{alg:CSPDSoS}
\begin{algorithmic}[1]
\Require $A_{l}\in\{0,1\}^{n\times n}$ for $l\in[L]$ and $K$.
\Ensure $\hat{\Pi}_{r}$ and $\hat{\Pi}_{c}$.
\State Set $S_{r}=\sum_{l\in[L]}(A_{l}A'_{l}-D^{r}_{l})$ and $S_{c}=\sum_{l\in[L]}(A'_{l}A_{l}-D^{c}_{l})$.
\State Obtain $\hat{U}_{r}$ and $\hat{U}_{c}$, the leading $K$ eigenvectors of $S_{r}$ and $S_{c}$, respectively.
\State Run SPA to $\hat{U}_{r}$'s (and $\hat{U}_{c}$'s) rows with $K$ clusters to get the estimated index set $\hat{\mathcal{I}_{r}}$ (and $\hat{\mathcal{I}_{c}}$).
\State Set $\hat{Y}_{r}=\mathrm{max}(0,\hat{U}_{r}\hat{U}^{-1}_{r}(\hat{\mathcal{I}_{r}},:))$ and $\hat{Y}_{c}=\mathrm{max}(0,\hat{U}_{c}\hat{U}^{-1}_{c}(\hat{\mathcal{I}_{c}},:))$.
\State Set $\hat{\Pi}_{r}(i,:)=\frac{\hat{Y}_{r}(i,:)}{\|\hat{Y}_{r}(i,:)\|_{1}}$ and $\hat{\Pi}_{c}(i,:)=\frac{\hat{Y}_{c}(i,:)}{\|\hat{Y}_{c}(i,:)\|_{1}}$ for $i\in[n]$.
\end{algorithmic}
\end{algorithm}

The computational cost of CSPDSoS is primarily attributed to steps 1, 2, and 3. Specifically, the time complexities of these steps are $O(n^{3}L)$, $O(n^{3})$, and $O(nK^{2})$, respectively. Given that $K\ll n$ in this article, the overall complexity of CSPDSoS is dominated by the $O(n^{3}L)$ term, rendering its total complexity as $O(n^{3}L)$.
\begin{rem}
Our CSPDSoS method, primarily designed to estimate asymmetric mixed memberships of nodes in multi-layer directed networks, can be effortlessly adapted for application in multi-layer bipartite networks. For example, consider a multi-layer bipartite network comprising $L$ layers, where $A^{bi}_{l}\in\{0,1\}^{n_{r}\times n_{c}}$ for $l\in[L]$, with $n_{r}$ and $n_{c}$ representing the number of row and column nodes, respectively \citep{qing2023community,qing2024bipartite}. Here, $A^{bi}_{l}$ denotes the bipartite adjacency matrix for the $l$-th layer. By merely substituting $A_{l}$ in Algorithm \ref{alg:CSPDSoS} with $A^{bi}_{l}$, we can estimate mixed memberships for both row and column nodes in multi-layer bipartite networks.
\end{rem}
\section{Main results}\label{sec5}
In this section, we establish per-node error bounds for the CSPDSoS algorithm's performance, demonstrating that the estimated membership matrices $\hat{\Pi}_{r}$ and $\hat{\Pi}_{c}$ converge closely to $\Pi_{r}$ and $\Pi_{c}$, respectively. To achieve this, we introduce the following assumption to govern $\mathcal{N}$'s sparsity.
\begin{assum}\label{Assum2}
$\rho^{2}n^{2}L\gg\mathrm{log}(n+L)$.
\end{assum}
Assumption \ref{Assum2} says that the sparsity parameter $\rho$ should diminish at a rate slower than $\sqrt{\frac{\mathrm{log}(n+L)}{n^{2}L}}$, aligning with the sparsity condition specified in Theorem 1 of \citep{lei2023bias}. Similar to Assumption 1.(b) in \citep{lei2023bias} and Assumption 2 in \citep{su2024spectral}, we also assume that the smallest singular values of $\sum_{l\in[L]}B_{l}B'_{l}$ and $\sum_{l\in[L]}B'_{l}B_{l}$ satisfy a linear growth respective to the layers $L$.
\begin{assum}\label{Assum22}
The smallest singular values of $\sum_{l\in[L]}B_{l}B'_{l}$ and $\sum_{l\in[L]}B'_{l}B_{l}$ are at least $c_{1}L$ and $c_{2}L$ for some constants $c_{1}>0$ and $c_{2}>0$, respectively.
\end{assum}
Analogous to Assumption 1(a) presented in \citep{lei2023bias} and the conditions outlined in Corollary 3.1 of \citep{mao2021estimating}, the following condition serves to facilitate our analysis.
\begin{con}\label{condition}
$K=O(1), \lambda_{K}(\Pi'_{r}\Pi_{r})=O(\frac{n}{K})$, and $\lambda_{K}(\Pi'_{c}\Pi_{c})=O(\frac{n}{K})$.
\end{con}
The theorem presented below serves as our key theoretical contribution, offering insights into the per-node error rates of the estimated mixed memberships delivered by our CSPDSoS algorithm.
\begin{thm}\label{MAIN}
If Assumptions \ref{Assum2}, \ref{Assum22}, and Condition \ref{condition} are satisfied, with probability $1-o(\frac{1}{n+L})$,
\begin{align*}
\mathrm{max}_{i\in[n]}\|e'_{i}(\hat{\Pi}_{r}-\Pi_{r}\mathcal{P}_{r})\|_{1}=O(\sqrt{\frac{\mathrm{log}(n+L)}{\rho^{2}n^{2}L}})+O(\frac{1}{n}), \mathrm{max}_{i\in[n]}\|e'_{i}(\hat{\Pi}_{c}-\Pi_{c}\mathcal{P}_{c})\|_{1}=O(\sqrt{\frac{\mathrm{log}(n+L)}{\rho^{2}n^{2}L}})+O(\frac{1}{n}),
\end{align*}
where $\mathcal{P}_{r}$ and $\mathcal{P}_{c}$ are two $K\times K$ permutation matrices.
\end{thm}
According to Theorem \ref{MAIN}, enhancing $\rho$ or increasing $n$ leads to more precise estimations of $\Pi_{r}$ and $\Pi_{c}$. Additionally, an augmentation in the number of layers $L$ results in a reduction in error rates, indicating that the utilization of multiple layers can significantly benefit the process of community detection.
\section{Numerical experiments}\label{sec6experiments}
\subsection{Synthetic data}
We first investigate CSPDSoS's performance on synthetic multi-layer directed networks by considering different values of the parameters $n, \rho$, and $L$.

\emph{Algorithms for comparison}. While \citep{su2024spectral} proposed an efficient method for community detection in multi-layer ScBM for non-overlapping multi-layer directed networks, it is unable to estimate nodes' mixed memberships in overlapping multi-layer directed networks. Meanwhile, as evidenced by Table \ref{LiteratureReview}, except for our CSPDSoS algorithm, as far as we know, there is no method specifically designed for estimating mixed memberships in overlapping multi-layer directed networks. Therefore, similar to \citep{lei2023bias, su2024spectral,qing2024bipartite}, we consider the following approaches for comparison:
\begin{itemize}
\item \texttt{CSPSum}: \textbf{C}o-clustering by \textbf{S}equential \textbf{P}rojection on the \textbf{Sum} of adjacency matrices. This method takes the leading $K$ left (and right) singular
vectors of $\sum_{l\in[L]}A_{l}$ to replace $\hat{U}_{r}$ (and $\hat{U}_{c}$) in Algorithm \ref{alg:CSPDSoS}.
\item \texttt{CSPSoS}: \textbf{C}o-clustering by \textbf{S}equential \textbf{P}rojection on the \textbf{S}um \textbf{o}f \textbf{S}quared matrices. CSPSoS takes $\sum_{l\in[L]}A_{l}A'_{l}$ and $\sum_{l\in[L]}A'_{l}A_{l}$ to replace $S_{r}$ and $S_{c}$ in Algorithm \ref{alg:CSPDSoS}, respectively.
\item \texttt{GeoNMF}, \texttt{SVM-cone-DCMMSB} (\texttt{SVM-cD} for short), and \texttt{Mixed-SCORE}: These three algorithms, initially presented in \citep{mao2017mixed}, \citep{maoSVM}, and \citep{jin2024mixed} respectively, are tailored for estimating nodes' mixed memberships in single-layer undirected networks. To adapt these algorithms for detecting overlapping communities in multi-layer directed networks, certain modifications are necessary. Inspired by our CSPDSoS algorithm and the DSoG algorithm in \citep{su2024spectral}, we use the aggregation matrices $S_{r}$ and $S_{c}$ to substitute for their adjacency matrix to estimate the row membership matrix $\Pi_{r}$ and the column membership matrix $\Pi_{c}$, respectively.
\end{itemize}
\begin{rem}\label{algorithmModify}
We should emphasize that the MASE \cite{arroyo2021inference}, SoS-Debias \citep{lei2023bias}, CAMSBM \citep{xu2023covariate}, and DSoG \citep{su2024spectral} algorithms are tailored for detecting nodes' non-overlapping communities in multi-layer networks but can't be adapted for estimating mixed memberships in such networks. Thus, we exclude them and their variants from our comparison. To our knowledge, the three approaches CSPDSoS, CSPSum, and CSPSoS developed in this paper are the only algorithms specifically designed for estimating mixed memberships in multi-layer directed networks. Given that GeoNMF, SVM-cD, and Mixed-SCORE are initially designed for estimating nodes' mixed memberships in single-layer undirected networks, we've adapted them for multi-layer directed networks to provide a more comprehensive comparison with our CSPDSoS method.
\end{rem}

\emph{Evaluation Metrics}. When the true membership matrices $\Pi_{r}$ and $\Pi_{c}$ are known, we consider the following three metrics:
\begin{itemize}
  \item \texttt{Hamming error:} This metric is defined as $\mathrm{max}\left(\frac{\mathrm{min}_{\mathcal{P}\in \mathcal{G}}\|\hat{\Pi}_{r}\mathcal{P}-\Pi_{r}\|_{1}}{n},\frac{\mathrm{min}_{\mathcal{P}\in \mathcal{G}}\|\hat{\Pi}_{c}\mathcal{P}-\Pi_{c}\|_{1}}{n}\right)$, where $\mathcal{G}$ represents the set of all $K\times K$ permutation matrices. Consistent with the measure used in Theorem \ref{MAIN}, which establishes the theoretical upper bound for per-node error rates through the $l_{1}$ difference between true and estimated mixed memberships for both row and column clusters, the Hamming error ranges from 0 to 1, with lower values indicating better performance.

  \item \texttt{Relative error:} This metric is defined as $\mathrm{max}\left(\frac{\mathrm{min}_{\mathcal{P}\in \mathcal{G}}\|\hat{\Pi}_{r}\mathcal{P}-\Pi_{r}\|_{F}}{\|\Pi_{r}\|_{F}},\frac{\mathrm{min}_{\mathcal{P}\in \mathcal{G}}\|\hat{\Pi}_{c}\mathcal{P}-\Pi_{c}\|_{F}}{\|\Pi_{c}\|_{F}}\right)$. Unlike the Hamming error, which measures the $l_{1}$ difference, the Relative error assesses the relative $l_{2}$ difference between the true and estimated membership matrices for both row and column communities. It is worth noting that both Hamming and Relative errors have been used in \citep{maoSVM,mao2021estimating,qing2023regularized,qing2024bipartite,jin2024mixed,qing2023finding} to evaluate the discrepancy between the true and estimated mixed membership matrices. A deeper comprehension of Hamming error and Relative error can be found in \citep{qing2024bipartite}.

  \item \texttt{Overlapping Normalized Mutual Information (ONMI):} Introduced by \citep{lancichinetti2009detecting}, this metric extends NMI \citep{danon2005comparing,bagrow2008evaluating} to evaluate overlapping community detection. It ranges from 0 to 1, with higher values indicating better performance. Detailed computations of ONMI can be found in \citep{lancichinetti2009detecting, OCCAM,moscato2021survey}, which we omit here. It is important to note that ONMI applies only to binary membership matrices \citep{maoSVM, OCCAM}. However, the membership matrices $\Pi_{r}, \Pi_{c}, \hat{\Pi}_{r}$, and $\hat{\Pi}_{c}$ considered in this study are continuous. Therefore, to utilize ONMI, we convert these matrices into their binary counterparts. For $\Pi_{r}$, we define its binary version as $\Pi^{0}_{r}$, where $\Pi^{0}_{r}(i,k)=1$ if $\Pi_{r}(i,k)\geq\delta$ and 0 otherwise, for $i\in[n]$ and $k\in[K]$. Here, $\delta$ is a threshold value set to $\frac{1}{K}$ in this study. Similarly, we define $\Pi^{0}_{c}, \hat{\Pi}^{0}_{r}$, and $\hat{\Pi}^{0}_{c}$ as the binary versions of $\Pi_{c}, \hat{\Pi}_{r}$, and $\hat{\Pi}_{c}$, respectively. Let $\mathrm{ONMI}(\Pi^{0}_{r}, \hat{\Pi}^{0}_{r})$ represent the ONMI between $\Pi^{0}_{r}$ and $\hat{\Pi}^{0}_{r}$, and $\mathrm{ONMI}(\Pi^{0}_{c}, \hat{\Pi}^{0}_{c})$ represent the ONMI between $\Pi^{0}_{c}$ and $\hat{\Pi}^{0}_{c}$. To simultaneously assess the true and estimated binary membership matrices for both row and column communities, we define the ONMI used in this paper as the minimum of $\mathrm{ONMI}(\Pi^{0}_{r}, \hat{\Pi}^{0}_{r})$ and $\mathrm{ONMI}(\Pi^{0}_{c}, \hat{\Pi}^{0}_{c})$, given that a higher ONMI value indicates better performance.
\end{itemize}

\begin{figure}
\centering
\resizebox{\columnwidth}{!}{
\subfigure[Experiment 1]{\includegraphics[width=0.3\textwidth]{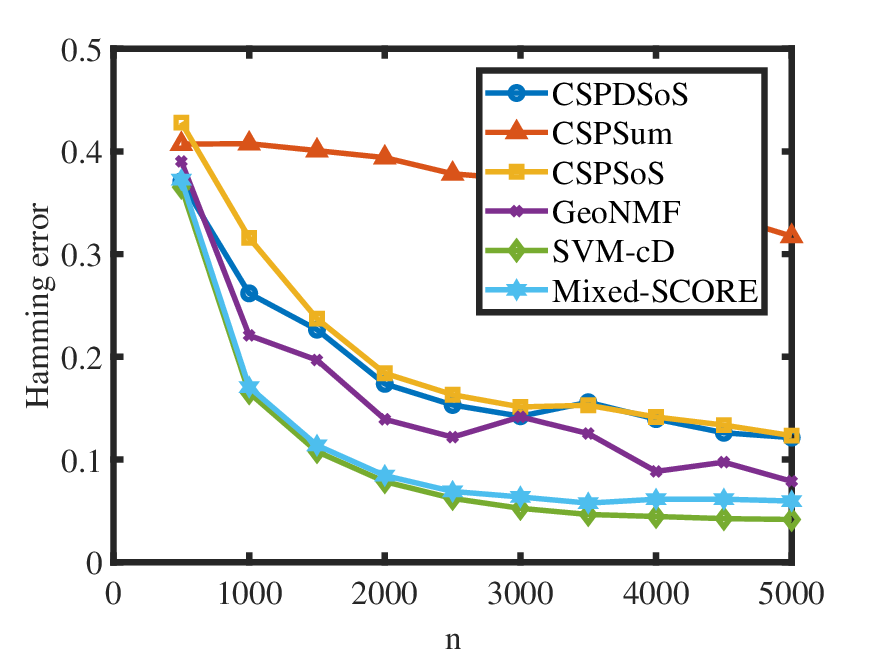}}
\subfigure[Experiment 1]{\includegraphics[width=0.3\textwidth]{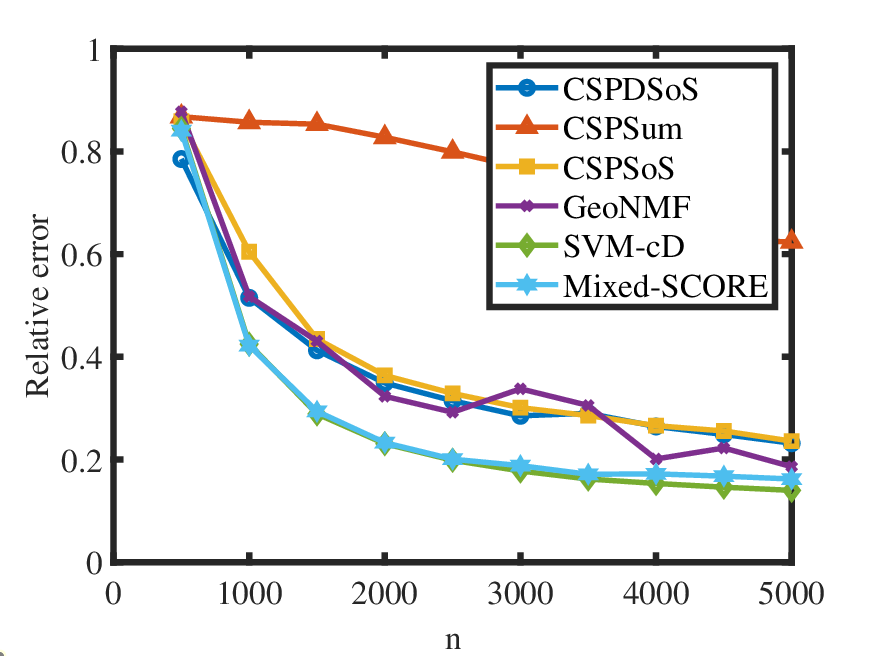}}
\subfigure[Experiment 1]{\includegraphics[width=0.3\textwidth]{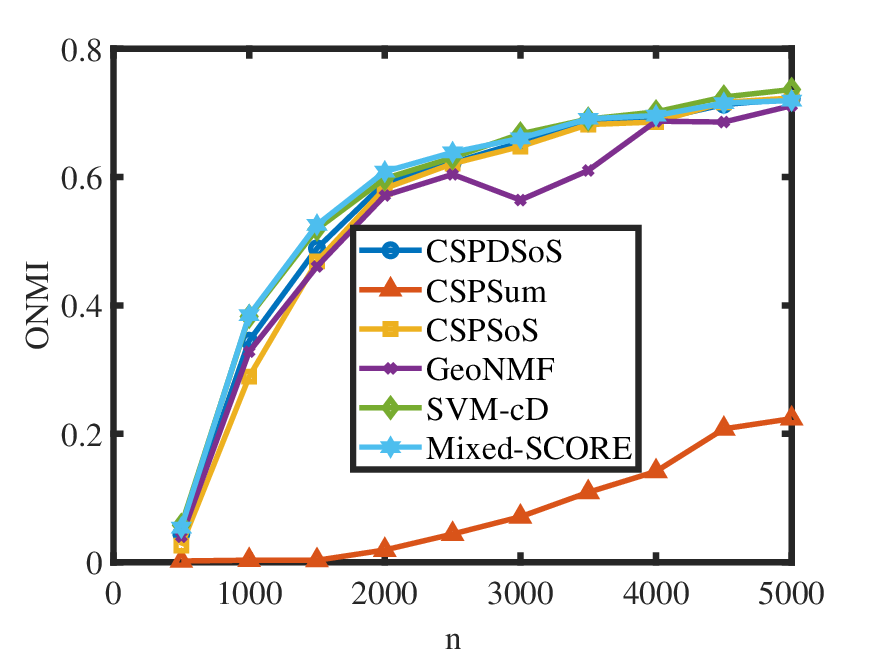}}
\subfigure[Experiment 1]{\includegraphics[width=0.3\textwidth]{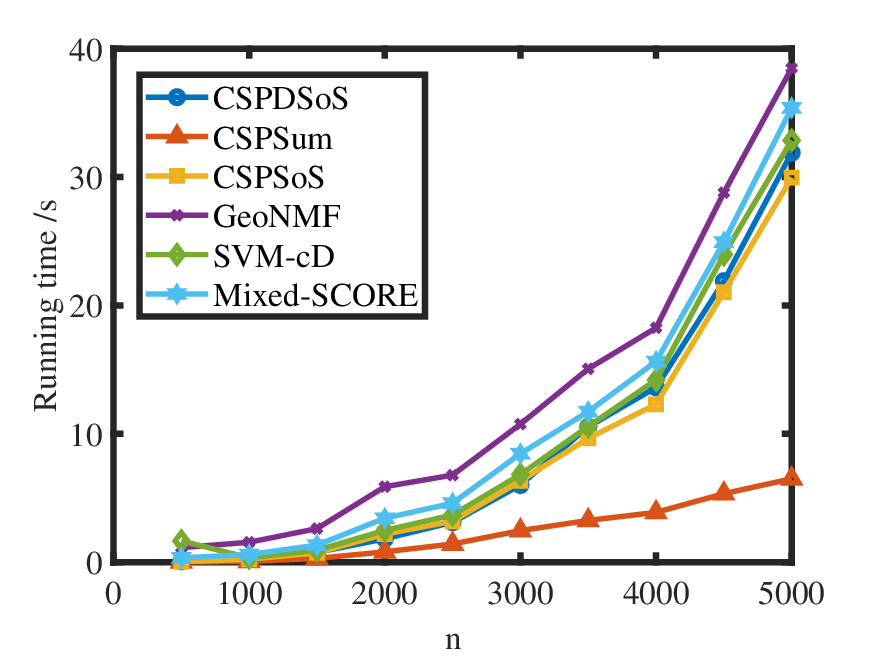}}
}
\resizebox{\columnwidth}{!}{
\subfigure[Experiment 2]{\includegraphics[width=0.3\textwidth]{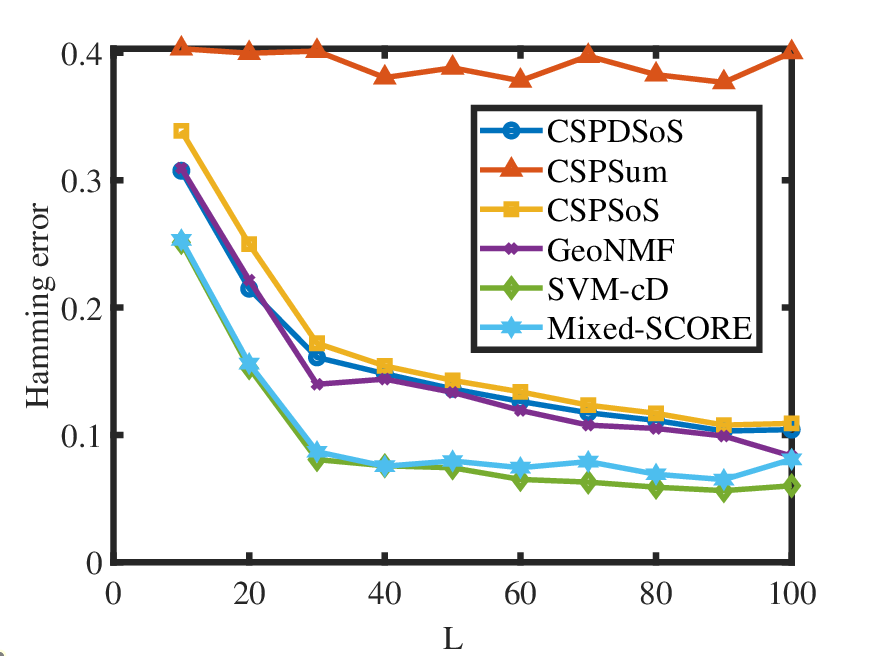}}
\subfigure[Experiment 2]{\includegraphics[width=0.3\textwidth]{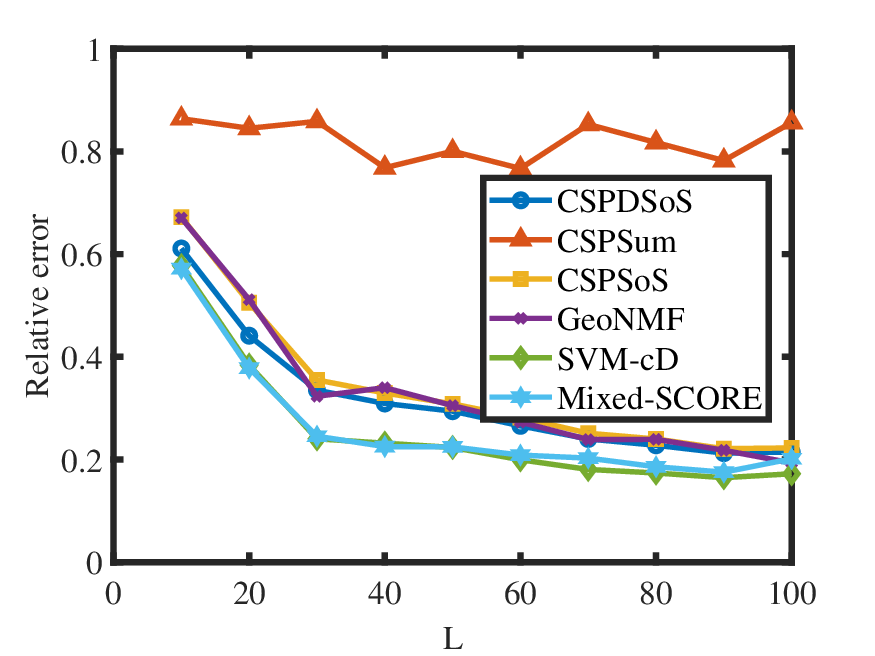}}
\subfigure[Experiment 2]{\includegraphics[width=0.3\textwidth]{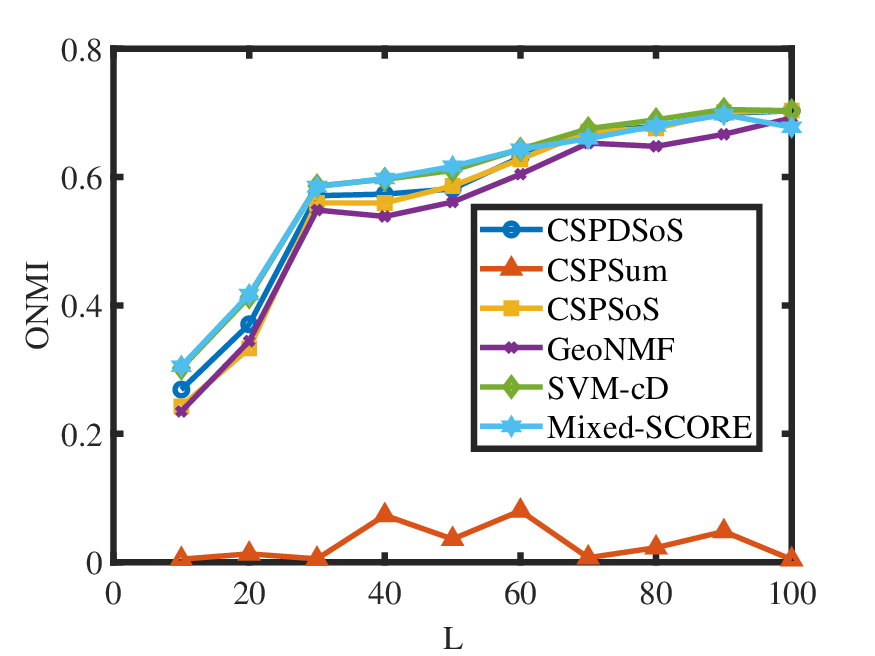}}
\subfigure[Experiment 2]{\includegraphics[width=0.3\textwidth]{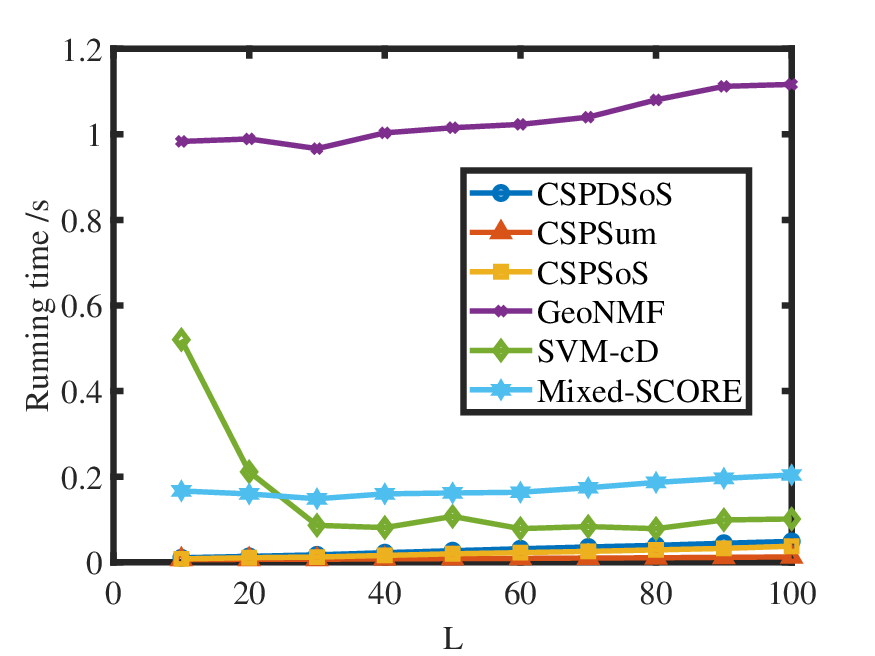}}
}
\resizebox{\columnwidth}{!}{
\subfigure[Experiment 3]{\includegraphics[width=0.3\textwidth]{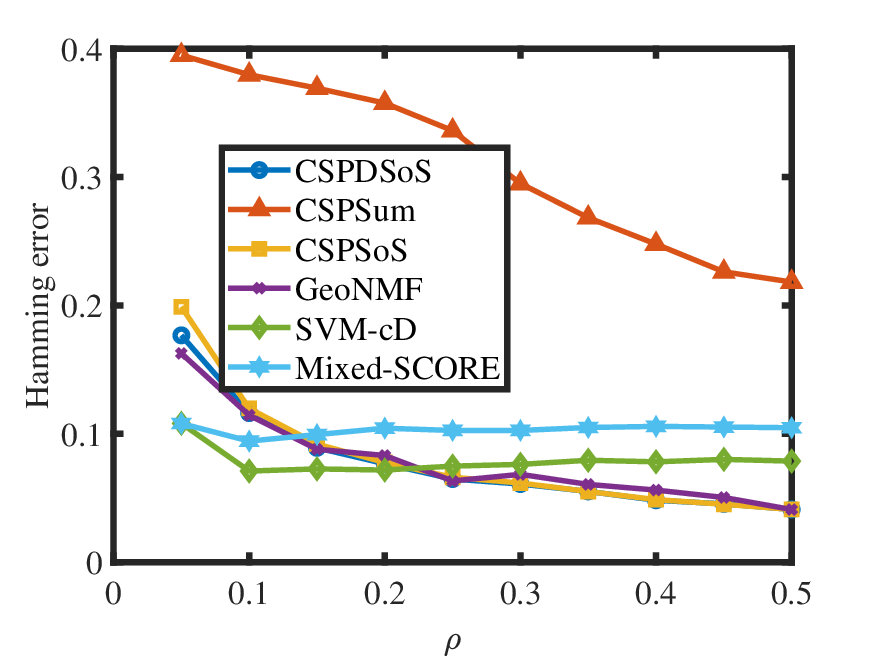}}
\subfigure[Experiment 3]{\includegraphics[width=0.3\textwidth]{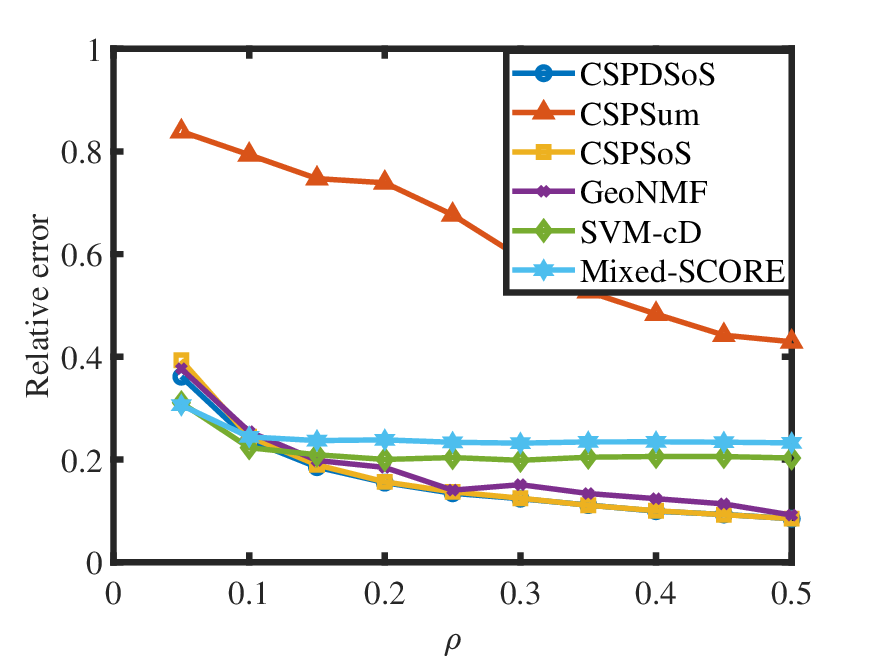}}
\subfigure[Experiment 3]{\includegraphics[width=0.3\textwidth]{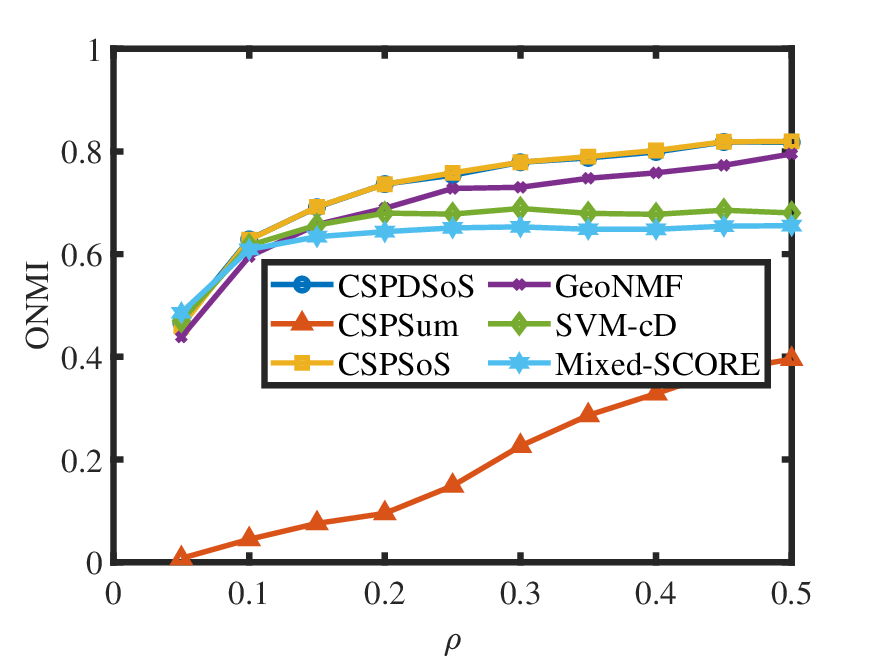}}
\subfigure[Experiment 3]{\includegraphics[width=0.3\textwidth]{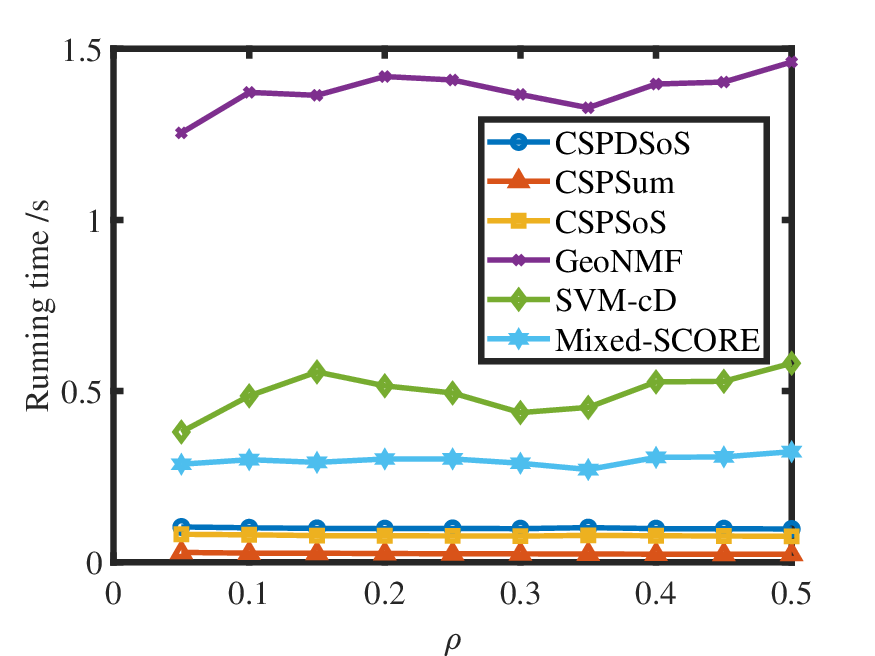}}
}
\caption{Numerical results of Experiments 1-3.}
\label{Ex} 
\end{figure}

\emph{Simulation configurations}. For synthetic multi-layer directed networks, we specify $K=3$ and denote $n^{r}_{0}$ and $n^{c}_{0}$ as the respective number of pure nodes in each row and column community. For the mixed membership of a mixed node $i$ belonging to the row community, we assign $r_{1}=\frac{\mathrm{rand}(1)}{2}$ , $r_{2}=\frac{\mathrm{rand}(1)}{2}$, and $r_{3}=1-r_{1}-r_{2}$, where $\mathrm{rand}(1)$ is a random value in $[0,1]$. Then we set $\Pi_{r}(i,:)=(r_{1},r_{2},r_{3})$. A similar approach is employed to generate the mixed memberships of nodes in the column community. Each entry of $B_{l}$ for $l\in[L]$ is randomly drawn from $[0,1]$. For each experiment, $n,L,\rho, n^{r}_{0}$, and $n^{c}_{0}$ are set independently. Finally, we report the average of each metric, obtained by conducting 100 independent repetitions under each set of configurations.

\emph{Experiment 1.} Here, our goal is to explore how the performance of these methods evolves as the number of nodes increases. To achieve this, we set $L=10$, $\rho=0.02$, $n^{r}_{0}=\frac{n}{4}$, $n^{c}_{0}=\frac{n}{5}$, and vary $n$ within the set $\{500, 1000, 1500, \ldots, 5000\}$.

\emph{Experiment 2.} The objective of this experiment is to assess the impact of increasing the number of layers $L$ in multi-layer directed networks on the methods' effectiveness. For this, we maintain constant values of $n=200$, $\rho=0.1$, $n^{r}_{0}=50$, $n^{c}_{0}=40$, while varying $L$ from the set $\{10, 20, 30, \ldots, 100\}$.

\emph{Experiment 3.} In this experiment, we seek to investigate the influence of the overall sparsity $\rho$ on the performance of the methods. We keep $n=500$, $L=20$, $n^{r}_{0}=80$, $n^{c}_{0}=120$ fixed, and test various values of $\rho$ from the range $\{0.05, 0.1, 0.15, \ldots, 0.5\}$.

\emph{Results.} The experimental results of Experiments 1-3 are displayed in Fig.~\ref{Ex} and they provide a rigorous validation of the proposed CSPDSoS algorithm for overlapping community detection in multi-layer directed networks, while also highlighting the performance of the CSPSoS, GeoNMF, SVM-cD, and Mixed-SCORE algorithms adapted for this task. Experiment 1 demonstrates that as the number of nodes $n$ increases, the CSPDSoS, CSPSoS, GeoNMF, SVM-cD, and Mixed-SCORE algorithms significantly outperform CSPSum in terms of Hamming error, Relative error, and ONMI, indicating more accurate estimation of the row and column membership matrices. This trend aligns with the theoretical prediction that increasing the number of nodes improves detection accuracy. Experiment 2 shows that increasing the number of layers $L$ in the multi-layer directed network further enhances the performance of CSPDSoS compared to the other algorithms, underscoring the value of additional layers in providing more information for community detection. Experiment 3 reveals that as the overall sparsity parameter $\rho$ decreases, representing denser networks, all algorithms achieve better performance, with lower Hamming and Relative errors, and higher ONMI. Particularly, our CSPDSoS enjoy competitive performances with CSPSoS, GeoNMF, SVM-cD, and Mixed-SCORE. This indicates that the proposed CSPDSoS algorithm is particularly effective in handling networks with denser connections. By panel (d) of this figure, we see that all methods are capable of detecting communities in a multi-layer directed network comprising 5000 nodes and 10 layers within 40 seconds. Overall, the consistent improvements across all three experiments not only confirm the theoretical analysis but also demonstrate the superior performance of the CSPDSoS, CSPSoS, SVM-cD, GeoNMF, and Mixed-SCORE methods for overlapping community detection in multi-layer directed networks.
\begin{rem}
Though GeoNMF, SVM-cD, and Mixed-SCORE have competitive performances with our CSPDSoS algorithm, it is important to note that we have modified the original GeoNMF, SVM-cD, and Mixed-SCORE algorithms from single-layer undirected networks to multi-layer directed networks, as done in Remark \ref{algorithmModify}. Therefore, the versions of GeoNMF, SVM-cD, and Mixed-SCORE used in this article are new methods tailored for multi-layer directed networks. Additionally, though these modified methods can detect mixed memberships for multi-layer directed networks generated from our multi-layer MM-ScBM model, we cannot obtain their theoretical analysis immediately. Instead, building their theoretical guarantees would require new models or techniques, which we leave as future work.
\end{rem}
\subsection{Real data}
\begin{table}[h!]
\centering
\caption{List of products considered in this article.}
\label{ProductsList}
\resizebox{\columnwidth}{!}{
\begin{tabular}{cccccccccccc}
\hline
``Soybeans''&``Food prep nes"&``Crude materials"&``Wine"&``Wheat"\\
``Oil, palm"&
``Meat, cattle, boneless (beef \&veal)"&``Cake, soybeans"& ``Beverages, distilled alcoholic"&``Maize'"\\
``Cheese, whole cow milk"&``Rubber natural dry"&``Cigarettes"&``Pastry"&``Chocolate products nes"\\
``Coffee, green"&``Meat, port"&``Meat, chicken"&``Sugar Raw Centrifugal"&``Cotton lint"\\
``Beverages, non alcoholic"&``Bananas"&``Rice, milled"&``Tobacco, unmanufactured"&``Fruit, prepared nes"\\
``Sugar refined"&``Beer of barley"&``Meat, pig"&``Cocoa, beans"&``Pet food"\\
\hline
\end{tabular}
}
\end{table}
We apply our CSPDSoS to the FAO Multiplex Trade Network (FAOMTN), a real-world multi-layer directed network accessible for download at \url{https://manliodedomenico.com/data.php} in this subsection. This dataset constitutes a comprehensive global network of food imports and exports, featuring multiple layers that represent different products. The nodes in this network denote individual countries, while the edges within each layer represent the import/export connections between countries for a particular product. The dataset, compiled in 2010, encompasses 364 products and 213 countries \citep{de2015structural}, where we have consolidated China and China-mainland into a single category labeled ``China" to ensure consistency. We have arranged the layers in decreasing order according to their total import/export value of the product and selected the top 30 products with the highest trading volume (i.e., $n=213, L=30$). For a comprehensive listing of the 30 products utilized in this study, please refer to Table \ref{ProductsList}. Meanwhile, Tables \ref{minImportExport} and \ref{maxImportExport} present the 25 countries with the smallest and largest import/export trade volumes, respectively, ranked based on their trade volume size among 213 countries. Initially, the data was weighted, but we have simplified it to an unweighted format whenever the trading value of the product between countries reaches or exceeds 100. After preprocessing the data, for the FAOMTN multi-layer directed network considered in this article, we have $A_{l}\in\{0,1\}^{213\times213}$ for $l\in[30]$.
\begin{table}[h!]
\centering
\caption{The 25 countries with the smallest import/export trade volumes. The numbers in parentheses represent the ranking of the trade volumes, where 1 indicates the smallest (i.e., ranked last among 213 countries) and 25 indicates the 25th smallest (i.e., ranked 189th among 213 countries).}
\label{minImportExport}
\resizebox{\columnwidth}{!}{
\begin{tabular}{cccccccccccc}
\hline
Mayotte (1)&Andorra (2)&Falkland Islands Malvinas (3)&Comoros (4)&Saint Pierre and Miquelon (5)\\
Tuvalu (6)&Palau (7)&Kiribati (8)&Montserrat (9)&Nauru (10)\\                        Cook Islands (11)&Eritrea (12)&Marshall Islands (13)&Micronesia, Federated States of (14)&Niue (15)\\
Saint Helena (16)&Somalia (17)&Guam (18)&Samoa (19)&Brunei Darussalam (20)\\
Cayman Islands (21)&Grenada (22)&American Samoa (23)&British Virgin Islands (24)&Equatorial Guinea (25)\\
\hline
\end{tabular}
}
\end{table}

\begin{table}[h!]
\footnotesize
\centering
\caption{The 25 countries with the largest import/export trade volumes. The numbers in parentheses represent the ranking of the trade volumes, where 1 indicates the largest (i.e., ranked first among 213 countries) and 25 indicates the 25th largest (i.e., ranked 25th among 213 countries).}
\label{maxImportExport}
\begin{tabular}{cccccccccccc}
\hline
United States (1)&Germany (2)&China (3)&Netherlands (4)&France (5)\\
Brazil (6)&United Kingdom (7)&Italy (8)&Japan (9)&Indonesia (10)\\
Canada (11)&Belgium (12)&Spain (13)&Malaysia (14)&Argentina (15)\\
Mexico (16)&Russia (17)&Australia (18)&Thailand (19)&Denmark (20)\\
Poland (21)&India (22)&Korea, Republic of (23)&Switzerland (24)&Austria (25)\\
\hline
\end{tabular}
\end{table}

Since the actual row and column memberships of the data are unknown, it is not feasible to assess the efficacy of the six methods employed on synthetic data. Instead, we consider the estimated row and column membership matrices generated by our CSPDSoS algorithm as the ground truth. By computing the Hamming error, Relative error, and ONMI for the remaining five methods as the number of communities increases, we aim to investigate how closely these methods compare to our CSPDSoS algorithm in analyzing this real-world dataset. Fig.~\ref{HROFAO} displays the results. We observe that as the number of communities $K$ increases, the estimated mixed memberships from the other five methods diverge more significantly from the estimated mixed membership matrices obtained from our CSPDSoS algorithm.

\begin{figure}
\centering
{\includegraphics[width=0.25\textwidth]{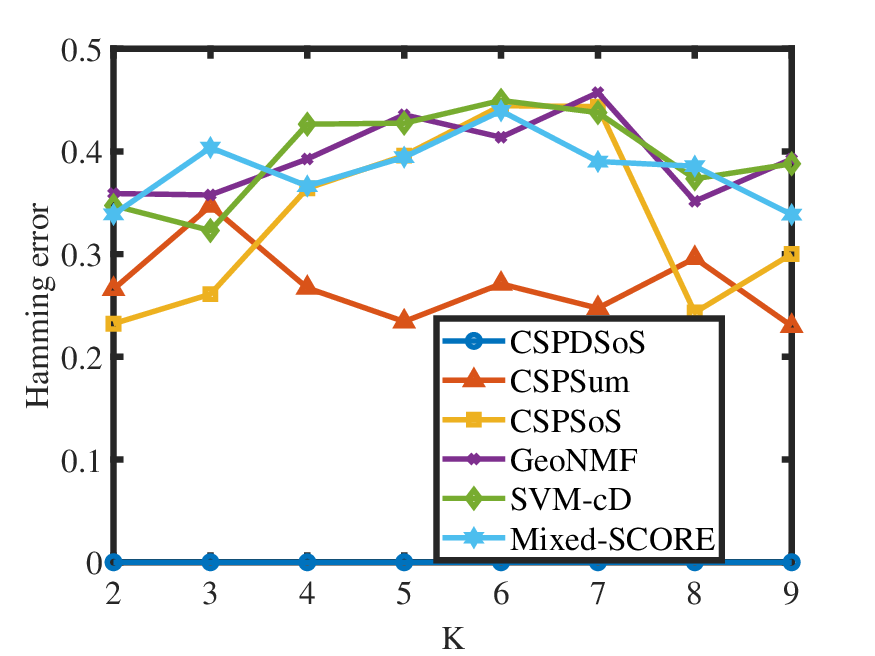}}
{\includegraphics[width=0.25\textwidth]{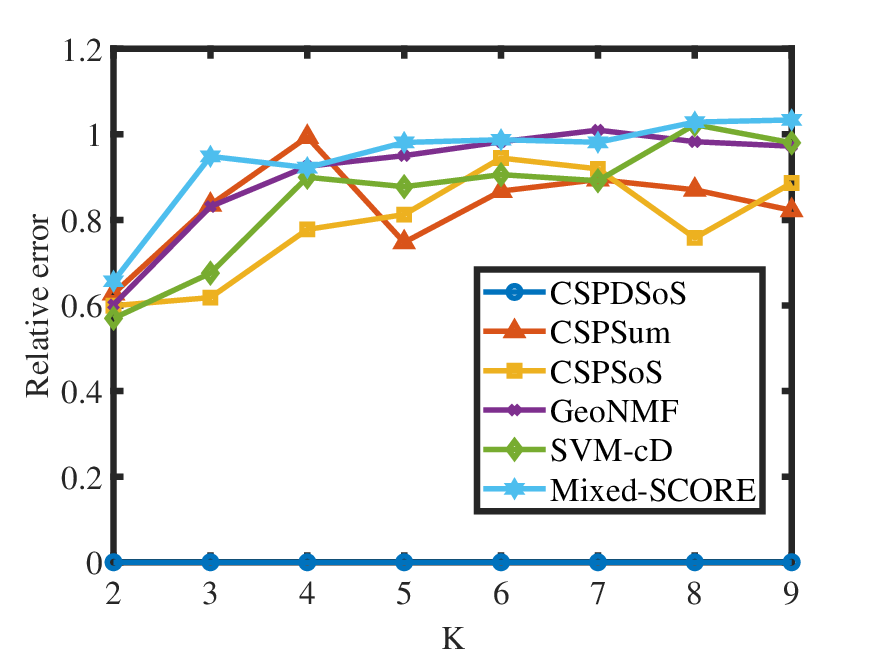}}
{\includegraphics[width=0.25\textwidth]{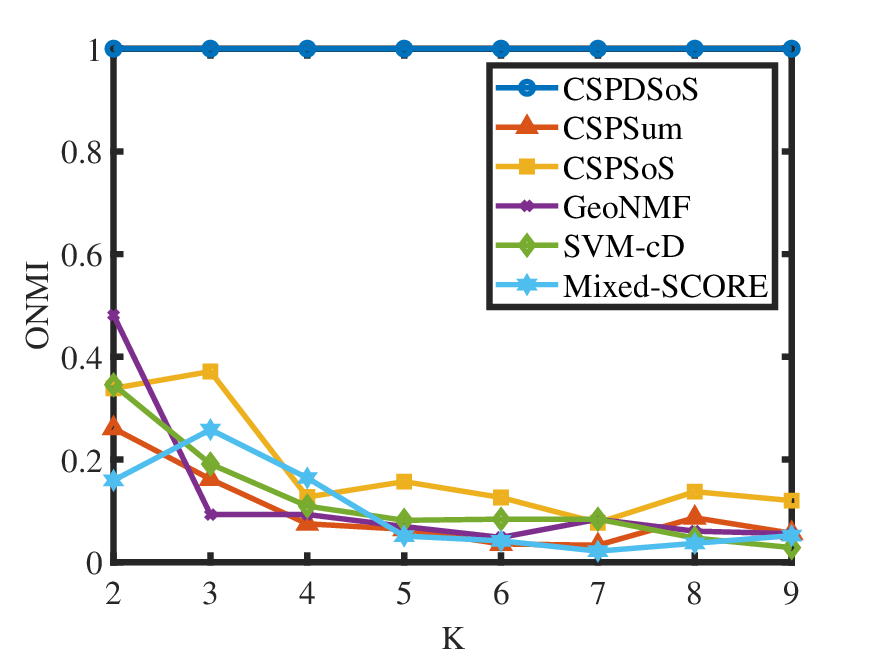}}
\caption{Comparison of different methods on the FAOMTN dataset using estimated row and column membership matrices from our CSPDSoS algorithm as ground truth. Plots show the Hamming error, Relative error, and ONMI for the remaining five methods as the number of communities increases.}
\label{HROFAO} 
\end{figure}
\begin{figure}
\centering
{\includegraphics[width=0.25\textwidth]{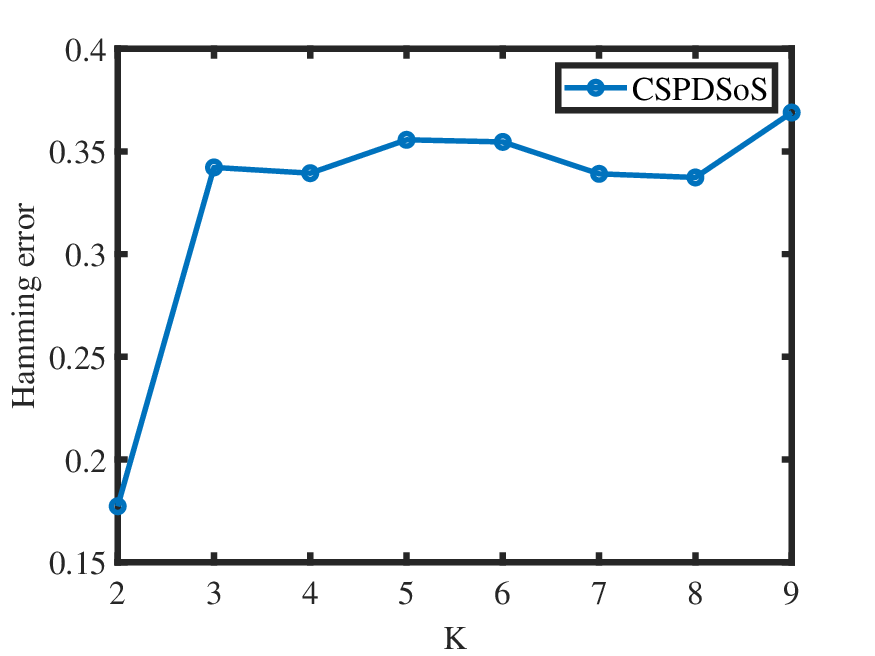}}
{\includegraphics[width=0.25\textwidth]{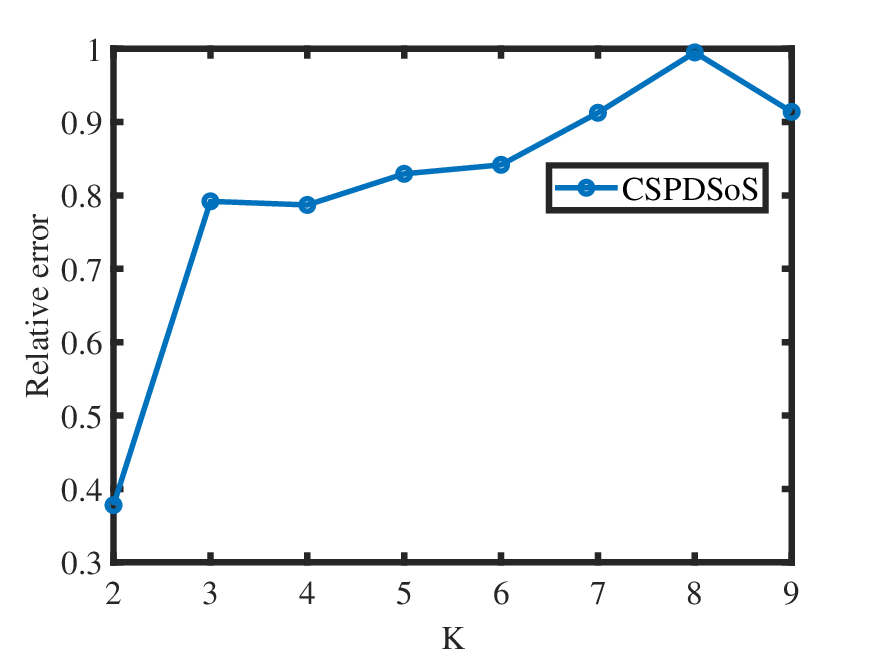}}
{\includegraphics[width=0.25\textwidth]{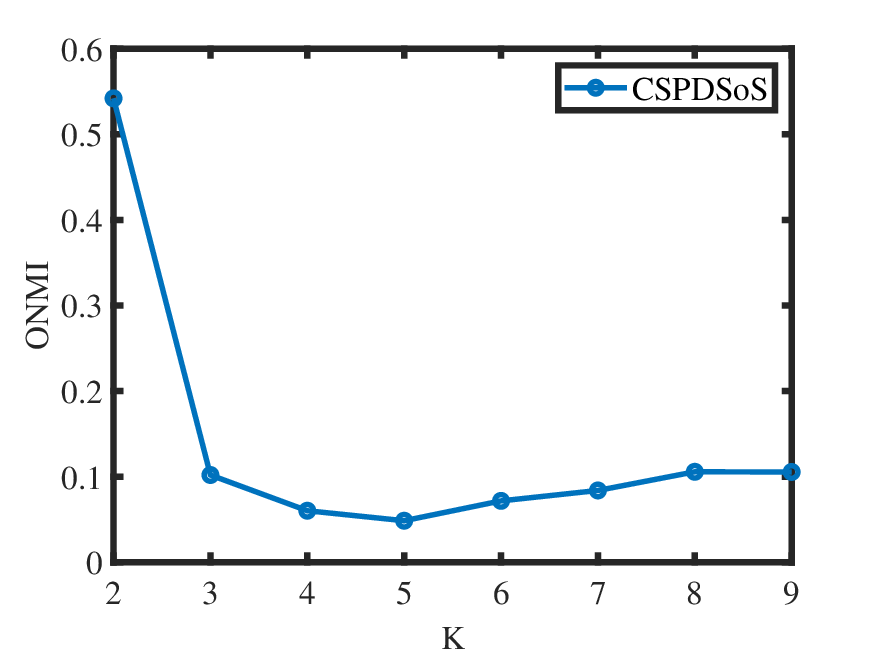}}
\caption{Overall asymmetric degree measured by Hamming error, Relative error, and ONMI between $\hat{\Pi}_{r}$ and $\hat{\Pi}_{c}$ returned by our CSPDSoS algorithm when $K$ increases for the FAOMTN dataset.}
\label{HROFAOCSPDSoS} 
\end{figure}

Since the true mixed membership matrices for the FAOMTN datasets remain unknown, we rely solely on our CSPDSoS algorithm for subsequent analysis. Given that FAOMTN is a multi-layer directed network, we are particularly interested in determining whether there is a significant asymmetry in the row and column community memberships. To answer this question, we compare the discrepancies between $\hat{\Pi}_{r}$ and $\hat{\Pi}_{c}$ generated by our CSPDSoS algorithm as the number of communities $K$ increases. This allows us to investigate the overall asymmetry between row and column communities. Fig.~\ref{HROFAOCSPDSoS} presents the results. It is evident that there is a notable asymmetry between the row and column mixed communities in the FAOMTN network, as indicated by the large Hamming and Relative errors, coupled with the relatively small ONMI values observed between $\hat{\Pi}_{r}$ and $\hat{\Pi}_{c}$. Recall that FAOMTN is a worldwide product import/export network, results in Fig.~\ref{HROFAOCSPDSoS} imply significant differences in the community structures between global import and export trades. From an economic network perspective, this asymmetry suggests that the clustering of countries into communities differs notably depending on whether one is looking at imports or exports.
\begin{figure}
\centering
\includegraphics[width=0.6\textwidth]{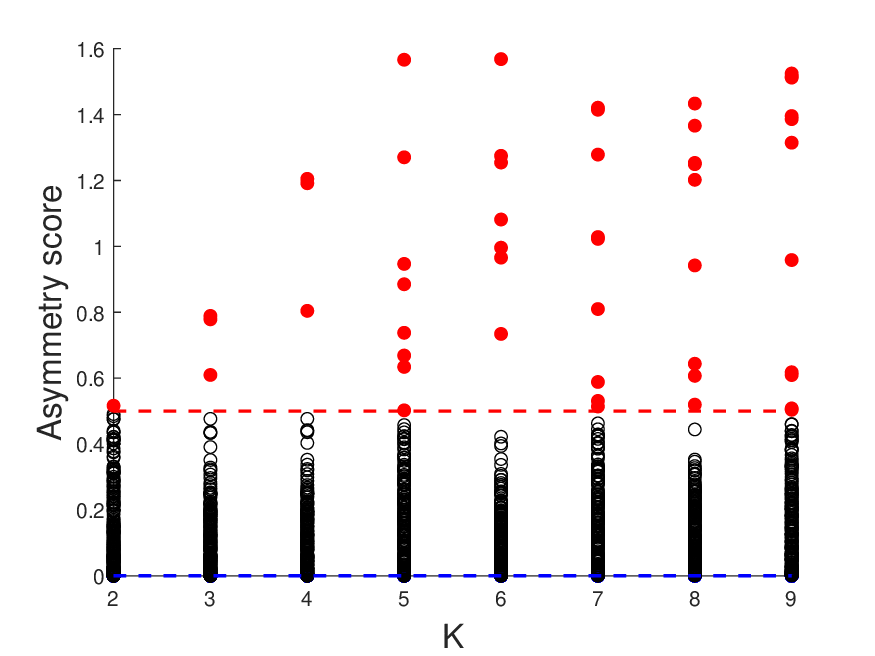}
\caption{Asymmetry score against the number of communities $K$ for the FAOMTN dataset. The red and blue dashed lines on the y-axis indicate the thresholds of $y=0.5$ and $y=0.0001$, respectively. Nodes with Asymmetry scores exceeding 0.5 are denoted by red solid circles, those with Asymmetry scores greater than 0.0001 but less than or equal to 0.5 are represented by black hollow circles, and those with Asymmetry scores less than or equal to 0.0001 are depicted by blue solid circles.}
\label{ASscoreFAO} 
\end{figure}

However, measuring the difference between $\hat{\Pi}_{r}$ and $\hat{\Pi}_{c}$ can only evaluate the overall asymmetric between row and column mixed communities. To find which node has the largest asymmetric pattern, we adopt the index Asymmetry score introduced in \citep{rohe2016co}. For $i\in[n]$, this index is defined as below:
\begin{align*}
a_{i}(K)=\|\hat{U}_{r}(i,:)-\hat{U}_{c}(i,:)\|_{F},
\end{align*}
where $\hat{U}_{r}$ and $\hat{U}_{c}$ are eigenvector matrix of $S_{r}$ and $S_{c}$, respectively. If the network is undirected, the Asymmetry score for any node is 0. A node with a high Asymmetry score exhibits distinct sending and receiving patterns \citep{rohe2016co}. Based on the results presented in Fig.~\ref{ASscoreFAO}, we have the following observations:
\begin{itemize}
  \item Only a few countries have Asymmetry scores no larger than 0.0001, indicating no significant difference between their import and export trades. In particular, we find 14 countries with an Asymmetry score no larger than 0.0001. These countries include Comoros, Nauru, Eritrea, Marshall Islands, Falkland Islands Malvinas, Cook Islands, Montserrat, Saint Pierre and Miquelon, Micronesia, Federated States of, Palau, Kiribati, Tuvalu, Mayotte, Andorra. This suggests that these countries' import and export activities are relatively balanced, possibly indicating economic self-sufficiency or a match between the demand for their imported and exported goods. Alternatively, it could indicate that these countries have minimal or almost no import/export relationships with other countries because we observe that all these 14 countries are listed in Table \ref{minImportExport}, i.e., their trade volumes are among the smallest in the 213 countries considered.
  \item For the majority of countries, the Asymmetry scores fall within the range of $(0.0001,0.5]$, suggesting a moderate difference between their import and export trades. This implies that while there is some disparity between what these countries import and export, the difference is not overwhelming. This moderate asymmetry may reflect a diversified economy where both imports and exports play a significant role.
  \item Only a few countries have Asymmetry scores that surpass 0.5, indicating a pronounced asymmetry between their import and export trades. Specifically, when $K=2$, it is exclusive to Germany; when $K$ is 3 or 4, this characteristic is observed in France, Germany, and the Netherlands; when $K=5$, it applies to Belgium, France, Germany, Italy, Netherlands, Poland, United Kingdom, and United States; when $K=6$, it is seen in Belgium, France, Germany, Italy, Netherlands, United Kingdom, and United States; when $K=7$, it includes Belgium, France, Germany, Italy, Netherlands, Poland, Spain, United Kingdom, and United States; when $K=8$, it pertains to Belgium, Canada, France, Germany, Italy, Netherlands, Spain, United Kingdom, and United States; and when $K=9$, it encompasses Belgium, Canada, China, France, Germany, Italy, Netherlands, Poland, Spain, United Kingdom, and United States. These countries are often major players in the global economy, with substantial import and export  volumes, and they may experience significant trade surpluses or deficits. This observation aligns with the data presented in Table \ref{maxImportExport}, as their trade volumes rank among the highest among all countries. It can be concluded that countries with small import/export volumes typically have very low Asymmetry scores, indicating no significant difference between their imports and exports, while countries with large import/export volumes often have high Asymmetry scores, indicating a significant difference between their imports and exports.
\end{itemize}

\begin{table}[h!]
\footnotesize
\centering
\caption{$\eta$ and $\sigma^{2}$ for all estimated communities for the FAOMTN network.}
\label{etasigma2}
\begin{tabular}{c|cccccc|cccccc}
\hline
&$C_{r,1}$&$C_{r,2}$&$C_{r,3}$&$C_{r,4}$&$C_{r,5}$&$C_{r,6}$&$C_{c,1}$&$C_{c,2}$&$C_{c,3}$&$C_{c,4}$&$C_{c,5}$&$C_{c,6}$\\
\hline
$\eta$&0.0719&0.0826&0.0598&0.0717&0.0832&0.6309&0.0364&0.0558&0.0833&0.1759&0.2473&0.4013\\
$\sigma^{2}$&0.0245&0.0223&0.0105&0.0131&0.0157&0.1691&0.0060&0.0119&0.0234&0.0632&0.0757&0.1285\\
\hline
\end{tabular}
\end{table}
To enhance interpretability, we have set $K=6$, representing the number of continents considered for trade analysis, excluding Antarctica for this particular study. We then apply our CSPDSoS approach to the FAOMTN dataset, assuming 6 row/column (import/export) communities. This results in the estimated mixed membership matrices $\hat{\Pi}_{r}\in[0,1]^{213\times6}$ for import trade and $\hat{\Pi}_{c}\in[0,1]^{213\times6}$ for export trade. Subsequently, we delve into a deeper analysis of these two  estimated mixed membership matrices. Recall that the $n\times 1$ membership vector $\hat{\Pi}_{r}(:,k)$ represents the extent to which each country belongs to the $k$-th estimated row (import) community $C_{r,k}$ for $k \in [K]$. We are interested in examining the likelihood of countries belonging to this community, as well as assessing the diversity within this community. We denote the mean of this vector by $\eta$ and its variance by $\sigma^{2}$. Specifically, we compute $\eta$ to measure the likelihood of countries belonging to this community, whereas $\sigma^{2}$ serves as an indicator of the diversity in these mixed memberships. Based on the results recorded in Table \ref{etasigma2}, we have the following observations:
\begin{itemize}
  \item Row (import) communities: The $\eta$ values range from 0.0598 to 0.6309. This indicates that some communities (e.g., $C_{r,3}$ with $\eta$=0.0598) have relatively low mean membership values, suggesting that these 213 countries are less likely to belong exclusively to these communities. In contrast, communities like $C_{r,6}$ ($\eta$=0.6309) have higher mean membership values, indicating that these countries are more likely to belong to this community. Meanwhile, the $\sigma^{2}$ values range from 0.0105 to 0.1691. Lower variance values (e.g.,$C_{r,3}$ with $\sigma^{2}=0.0105$) indicate that countries within the community have more homogeneous membership values, while higher variance values (e.g., $C_{r,6}$ with $\sigma^{2}=0.1691$) suggest higher membership diversity.
  \item Column (export) communities: The $\eta$ values range from 0.0364 to 0.4013. Similar to the import communities, some export communities have lower mean membership values (e.g., $C_{c,1}$ with $\eta$=0.0364), while others have higher values (e.g., $C_{c,6}$ with $\eta$=0.4013). Meanwhile, the $\sigma^{2}$ values range from 0.0060 to 0.1285. Similar to import communities, some export communities exhibit lower variance (e.g., $C_{c,1}$ with $\sigma^{2}=0.0060$), indicating homogeneity, while others show higher variance (e.g., $C_{c,6}$ with $\sigma^{2}=0.1285$), indicating diversity.
  \item Comparison between import and export Communities: The mean membership values $\eta$ for import communities are generally higher than those for export communities, suggesting that countries tend to have stronger affiliations with import communities. The variance in membership values $\sigma^{2}$ for import communities is also generally higher than for export communities, indicating greater diversity in membership affiliations within import communities. These findings suggest that import communities have stronger and more diverse affiliations among countries.
\end{itemize}

For country $i$, if $\mathrm{max}_{k\in[K]}\hat{\Pi}_{r}(i,k)\leq0.4$, it is classified as a highly mixed import country, suggesting a diversified economy with imports from various countries. Similarly, if $\mathrm{max}_{k\in[K]}\hat{\Pi}_{c}(i,k)\leq0.4$, it is deemed a highly mixed export country. Fig.~\ref{FAOMixedK6rc} shows the lists of these highly mixed import (export) countries and their estimated mixed memberships obtained from CSPDSoS. The results indicate that these countries exhibit low membership values across all import/export communities, suggesting a lack of strong preference for any specific one. This indicates a diversified economy, where they import and export various goods from multiple sources and regions. This pattern may arise from their small size or economic diversification strategies, which aligns with the observation that many of these countries are small and not heavily dependent on a particular good or region. Instead, they trade with multiple partners to satisfy their economic needs, and some may have implemented diversification strategies to minimize dependency on a single export or import. Fig.~\ref{FAOminK6rc} displays the mixed memberships for countries with the smallest trade volumes, as listed in Table \ref{minImportExport}. Like Fig.~\ref{FAOMixedK6rc}, the heatmap reveals that these countries have low membership values in all import and export communities.

Fig.~\ref{FAOmaxK6rc} displays the mixed memberships  for the top 25 countries with the largest import/export trade volumes in the FAOMTN dataset. The heatmaps reveal interesting patterns of membership distributions across different import (row) and export (column) communities. These patterns provide insights into the complex trade dependencies and integration of countries within the global economy. In detail, we have the following observations:
\begin{itemize}
  \item In the import communities, countries like the United States, Germany, Netherlands, France, Italy, India, China, Brazil, Japan, Indonesia, Malaysia, Argentina, and Mexico show strong membership values predominantly in one import community, suggesting these countries have dominant trade relationships with a particular economic bloc. In particular, the pure nodes are found to be United States, Germany, Netherlands, France, Italy, and India. On the other hand, countries like the United Kingdom, Canada, Belgium, Spain, Russia, Denmark, Poland, Switzerland, and Austria exhibit more balanced memberships across multiple import communities, indicating their trade relations are diversified and spread across various economic groups.
  \item Similarly, in the export communities, certain countries like Germany, Netherlands, France, Brazil, United Kingdom, Argentina, Mexico, Thailand, and India display concentrated memberships in one export community, reflecting their dominance in a specific export market. In particular, the pure nodes are found to be Germany, Netherlands, France, and United Kingdom. Conversely, countries such as the United States, China, Italy, Japan, Indonesia, Canada, Belgium, Spain, Malaysia, Russia, Australia, Denmark, Poland, Korea, Republic of, Switzerland, and Austria show more diverse distributed memberships across multiple export communities, suggesting they are active players in multiple export markets.
\end{itemize}

Let the term $\mathrm{argmax}_{k\in[K]}\hat{\Pi}_{r}(i,k)$ (and $\mathrm{argmax}_{k\in[K]}\hat{\Pi}_{c}(i,k)$) refer to the home base row (and column) community of node $i$ for $i\in[n]$. Fig.~\ref{FAOtop3RowColumncommunities} presents the home base community assignments returned by the CSPDSoS method for three selected layers, offering a clear visualization of the identified communities. We observe that (a) in the ``Soybeans" layer and the ``Crude material" layer, there are many highly mixed countries in the import communities, indicated by the prevalence of green squares, while there are only a few highly mixed countries in the export communities; (b) for the ``Food prep nes" layer, there are several highly mixed countries in both import and export communities; (c) across all three layers, the import communities differ significantly from the export communities, reflecting the asymmetric nature of trade relationships. Thus, this figure illustrates the complex and often asymmetric community structures in the trade network.

Fig.~\ref{FAOK6WorldMapR} and Fig.~\ref{FAOK6WorldMapC} present global maps of estimated mixed membership intensity in the six import communities and the six export communities for all 213 countries, respectively. These maps offer valuable insights into the dominant trade positions of countries within each import and export community. For example, in the import community $C_{r,1}$ (where the mixed memberships of countries belonging to this community are recorded in $\hat{\Pi}_{r}(:,1)$, the first column of $\hat{\Pi}_{r}$), it is evident that Germany stands out with the darkest color, signifying its dominant position in this community. This observation aligns with the results shown in Fig.~\ref{FAOmaxK6rc}, where Germany is identified as a pure node in $C_{r,1}$, further reinforcing its central role in this import community. The colors in the last panel of Fig.~\ref{FAOK6WorldMapR} (representing $C_{r,6}$) are generally darker than those in the other five panels, indicating that the membership values of countries belonging to $C_{r,6}$ are typically larger than those in the other import communities. This finding is consistent with the statistical summary in Table \ref{etasigma2}, which shows that the mean membership value ($\eta$) for $C_{r,6}$ is the highest among all import communities. Moreover, the maps reveal that countries with closer geographical locations tend to have similar membership intensities within each import community. For instance, European (or Asian) countries exhibit closely aligned memberships in all import communities, suggesting that these regions share similar import relationships and dependencies. A similar pattern emerges in the global map of export communities (Fig.~\ref{FAOK6WorldMapC}). Germany, Netherlands, France, and United Kingdom stand out as dominant players in various export communities. As with the import communities, the membership intensities of countries within export communities also tend to cluster geographically.

\begin{figure}
\centering
{\includegraphics[width=1\textwidth]{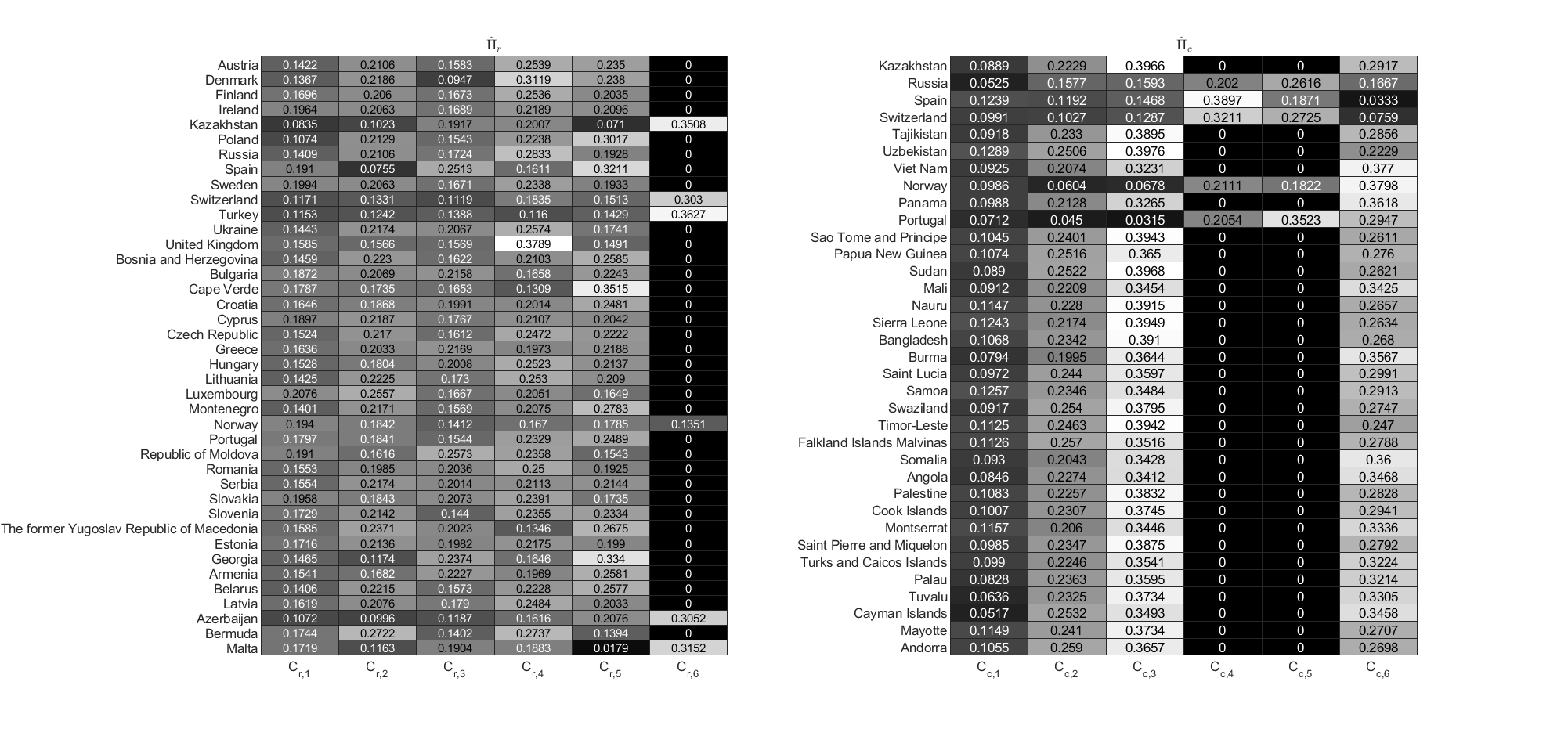}}
\caption{Heatmap of the estimated mixed memberships for highly mixed import/export countries.}
\label{FAOMixedK6rc} 
\end{figure}

\begin{figure}
\centering
{\includegraphics[width=1\textwidth]{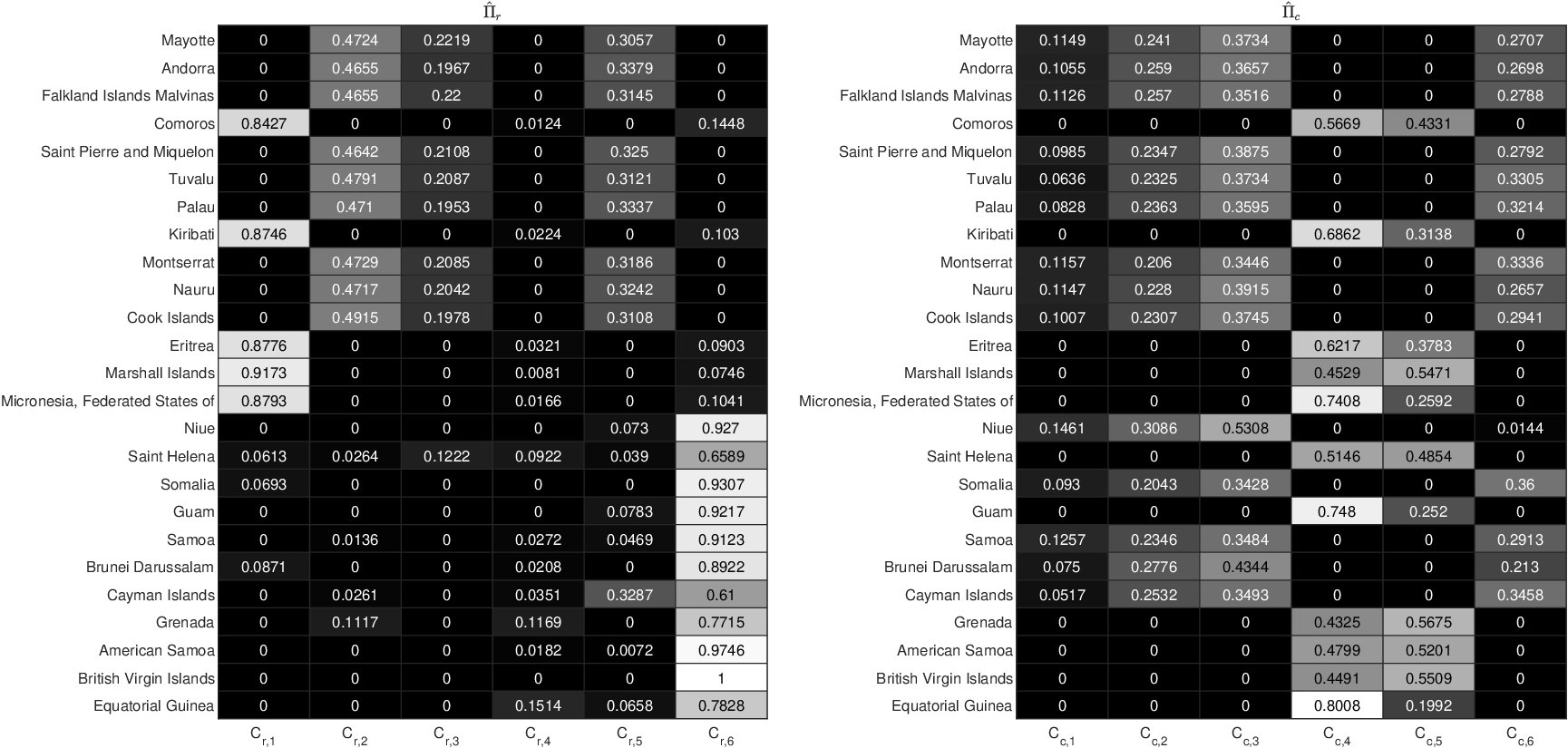}}
\caption{Heatmap of the estimated mixed memberships for countries in Table \ref{minImportExport}.}
\label{FAOminK6rc} 
\end{figure}

\begin{figure}
\centering
{\includegraphics[width=1\textwidth]{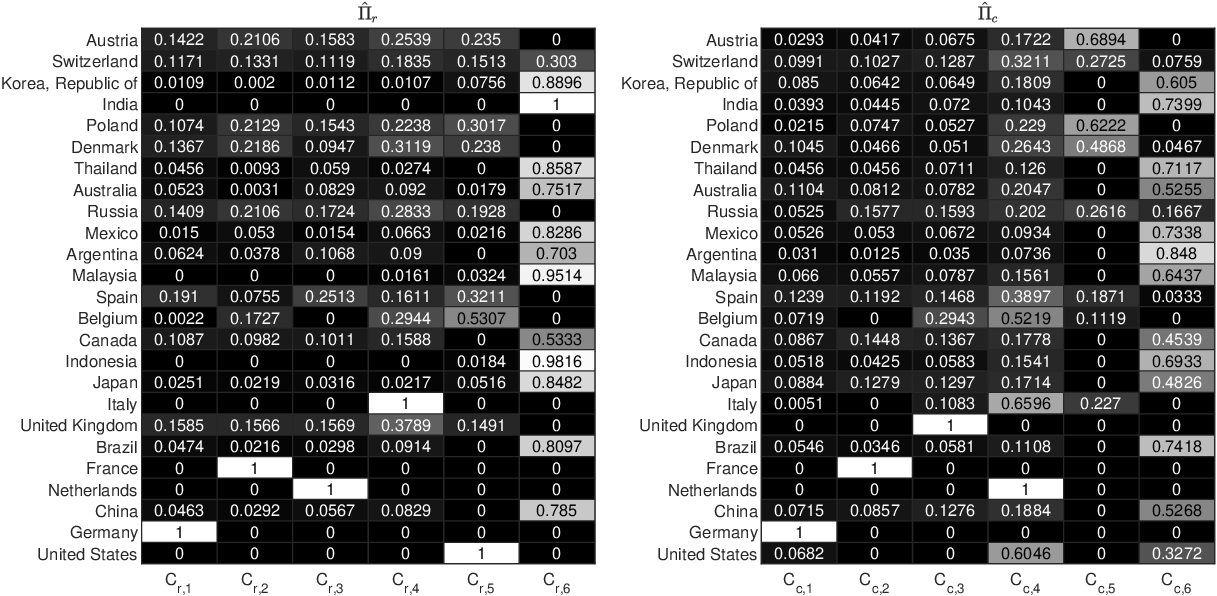}}
\caption{Heatmap of the estimated mixed memberships for countries in Table \ref{maxImportExport}.}
\label{FAOmaxK6rc} 
\end{figure}

\begin{figure}
\centering
\subfigure[Soybeans]{\includegraphics[width=0.485\textwidth]{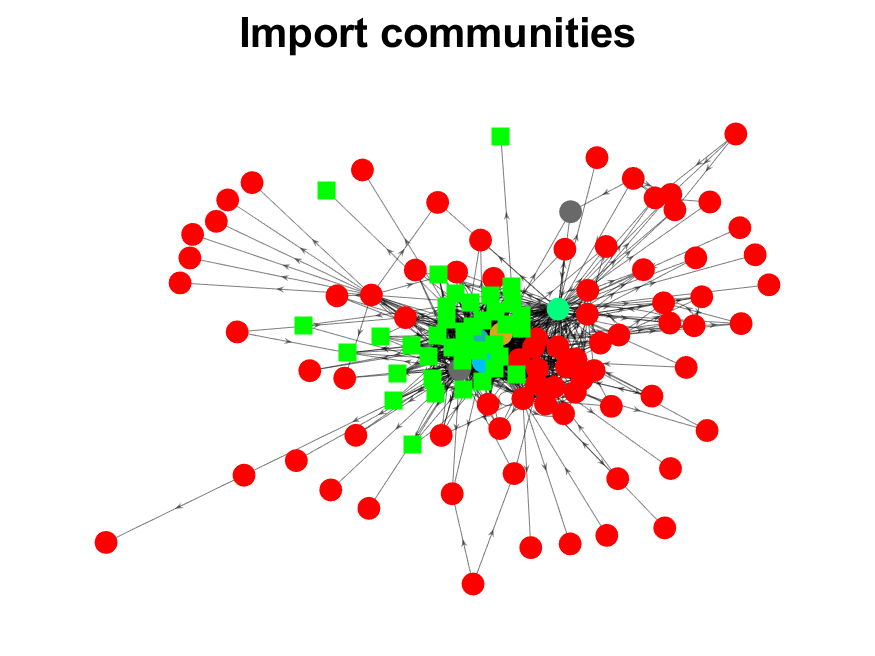}}
\subfigure[Soybeans]{\includegraphics[width=0.485\textwidth]{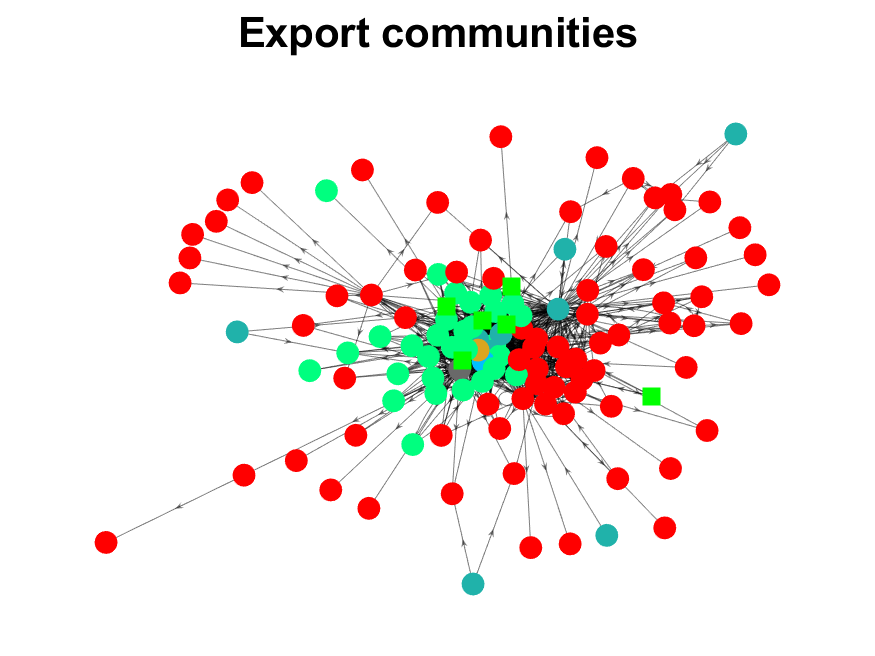}}
\subfigure[Food prep nes]{\includegraphics[width=0.485\textwidth]{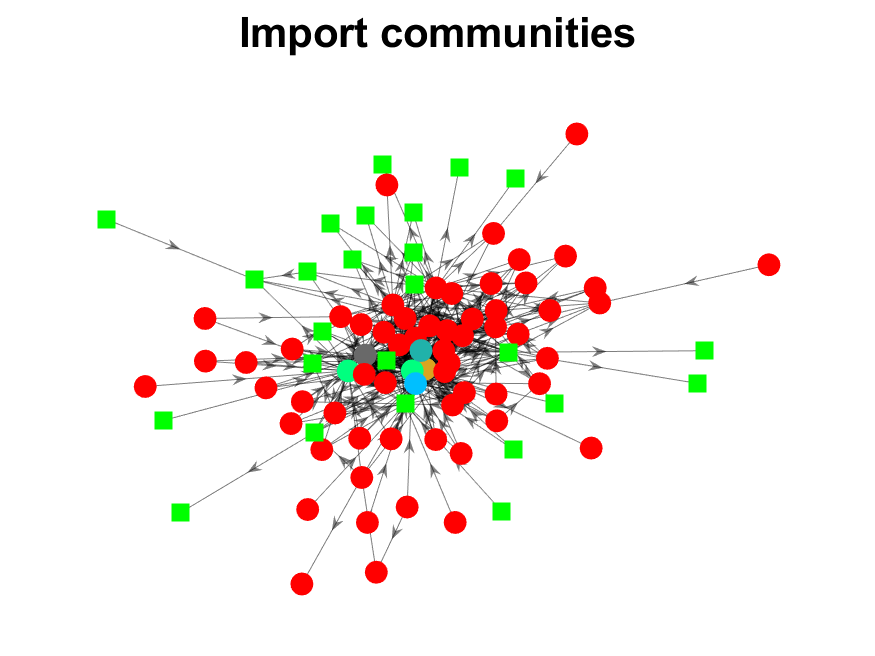}}
\subfigure[Food prep nes]{\includegraphics[width=0.485\textwidth]{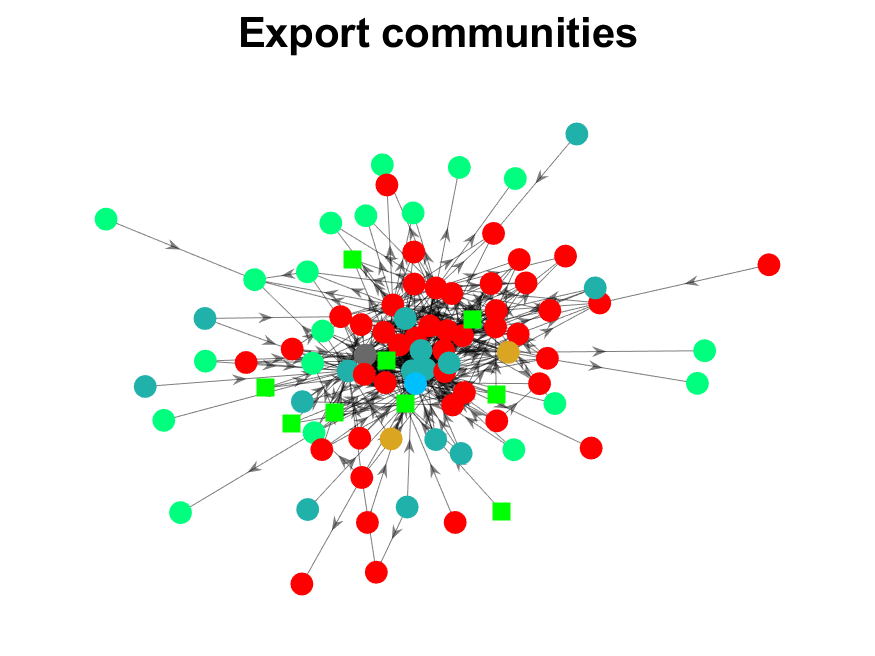}}
\subfigure[Crude material]{\includegraphics[width=0.485\textwidth]{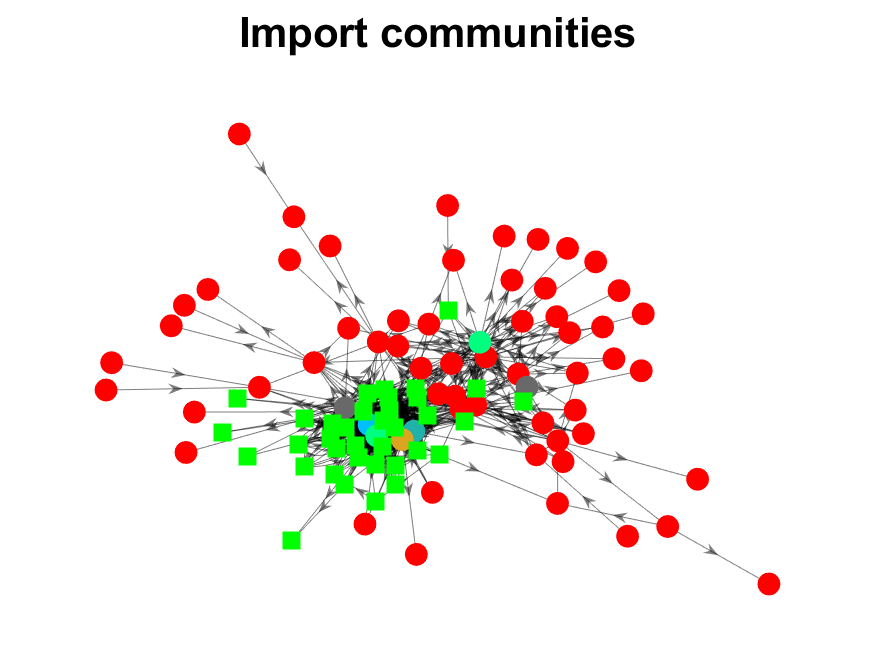}}
\subfigure[Crude material]{\includegraphics[width=0.485\textwidth]{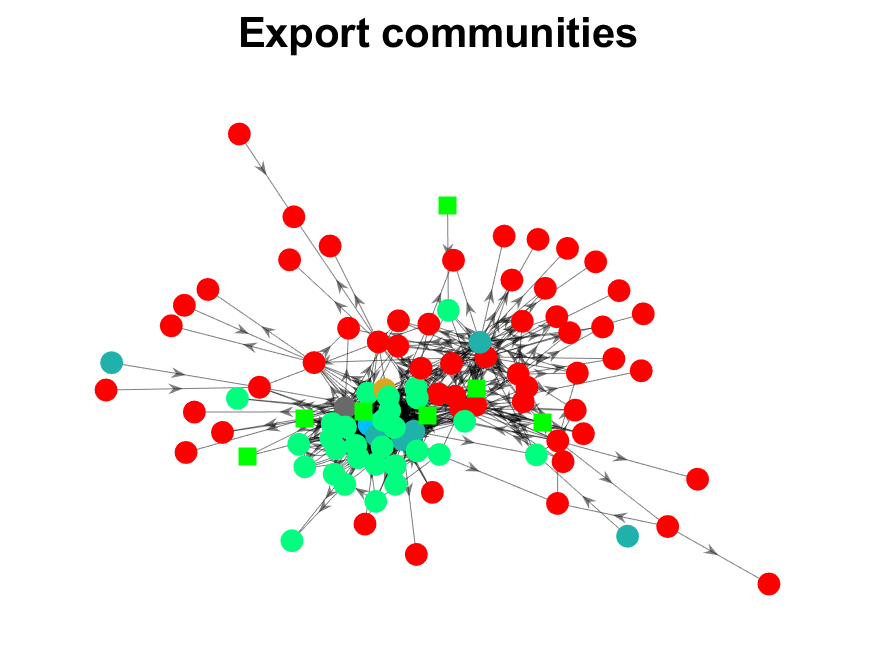}}
\caption{Examples of the home base community assignments for the three layers: Soybeans, Food prep nes, and Crude material. Colors indicate communities and green squares mean highly mixed nodes. For visualization, we do not show isolated nodes for the three layers.}
\label{FAOtop3RowColumncommunities} 
\end{figure}

\begin{figure}
\centering
{\includegraphics[width=0.495\textwidth]{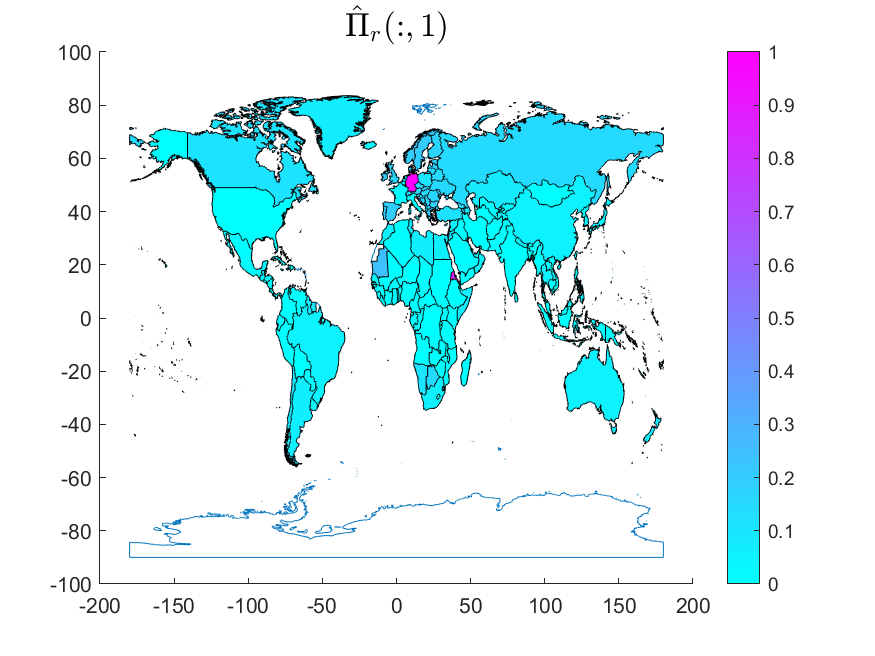}}
{\includegraphics[width=0.495\textwidth]{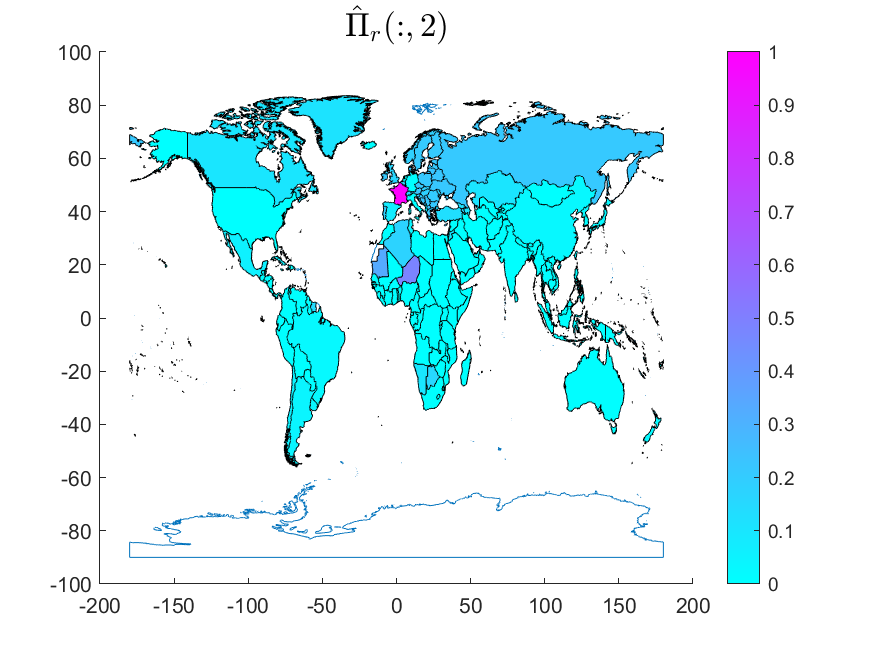}}
{\includegraphics[width=0.495\textwidth]{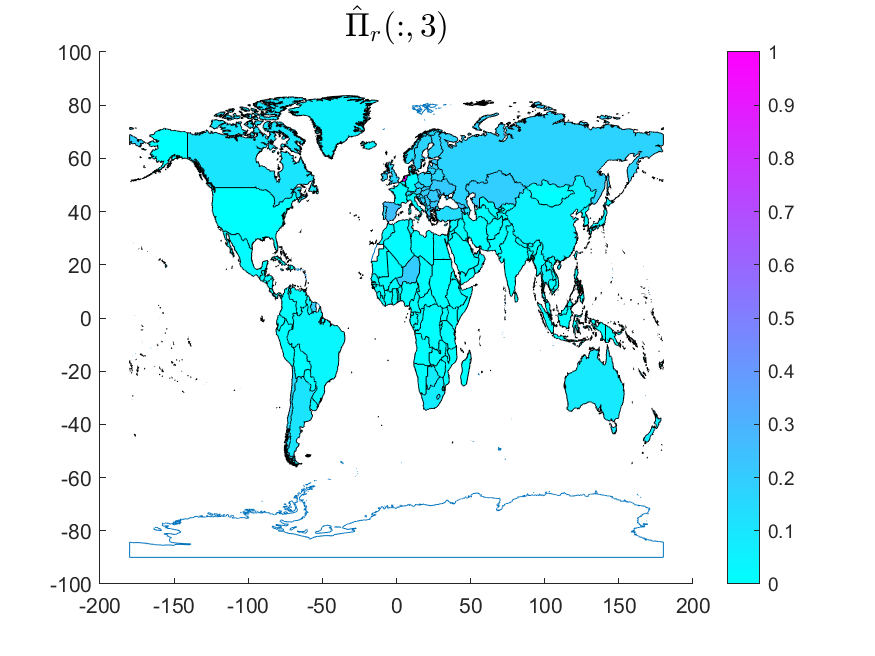}}
{\includegraphics[width=0.495\textwidth]{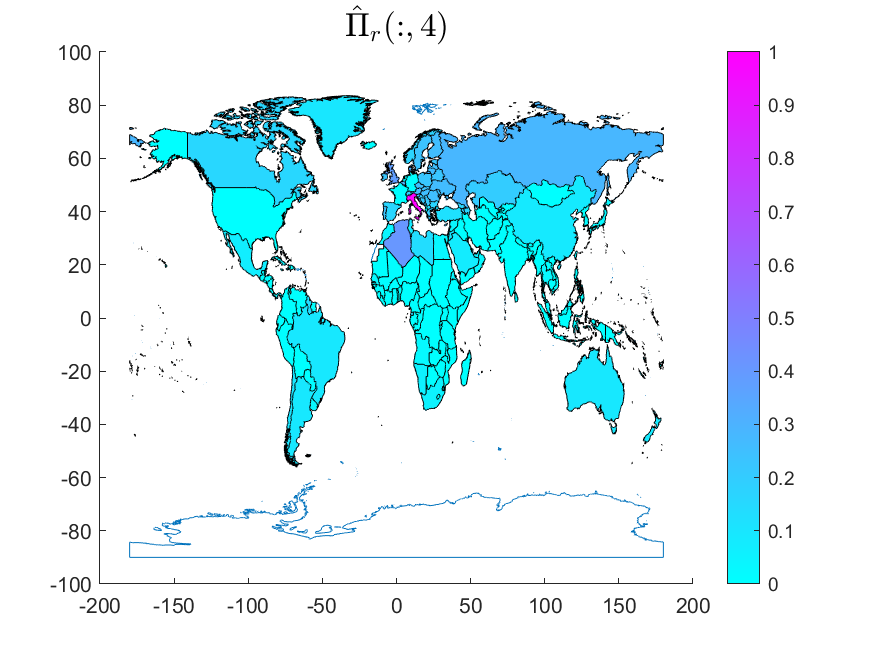}}
{\includegraphics[width=0.495\textwidth]{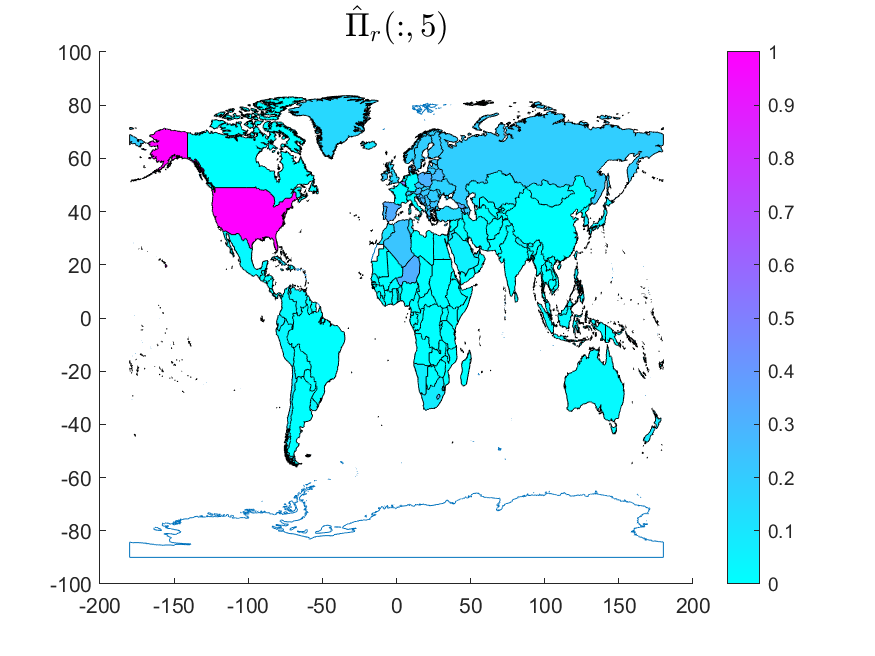}}
{\includegraphics[width=0.495\textwidth]{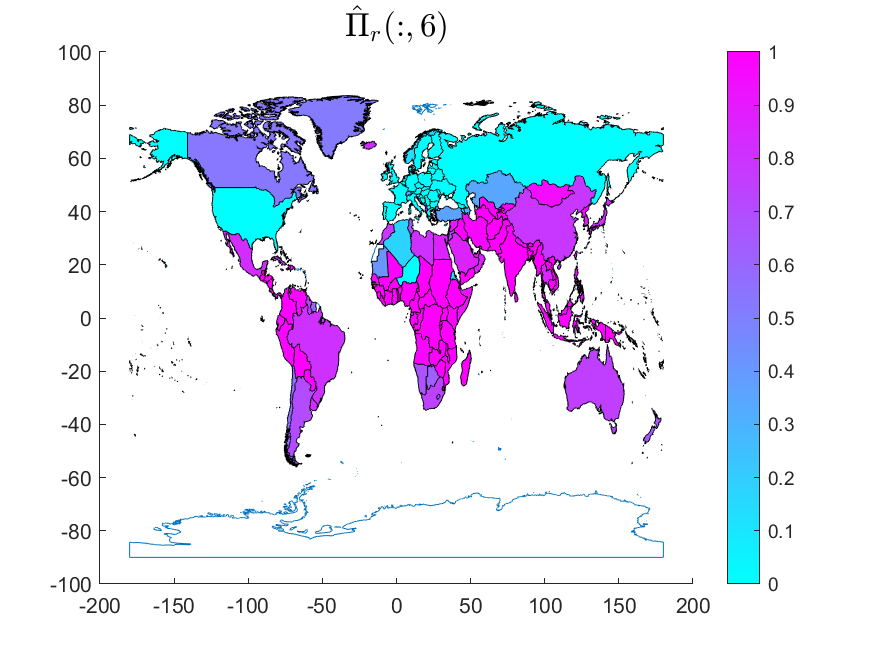}}
\caption{Global map of estimated mixed membership intensity in the six import communities for all the 213 countries. For example, the first sub-figure depicts the membership intensity of each country belonging to the first import community $C_{r,1}$ as estimated by our CSPDSoS algorithm. For $i\in[213]$, recall that the membership intensity for country $i$ in import community 1 is represented by $\hat{\Pi}_{r}(i,1)$ . The color of country $i$ on the map corresponds to the magnitude of $\hat{\Pi}_{r}(i,1)$, with darker colors indicating higher membership intensity and lighter colors indicating lower membership intensity. Similar arguments for the other five import communities.}
\label{FAOK6WorldMapR}
\end{figure}
\begin{figure}
\centering
{\includegraphics[width=0.495\textwidth]{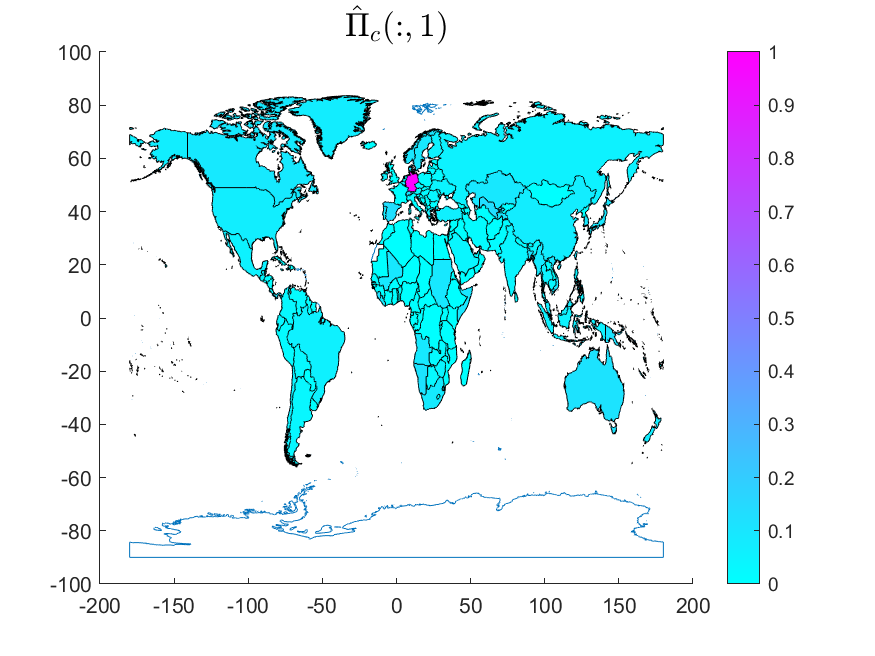}}
{\includegraphics[width=0.495\textwidth]{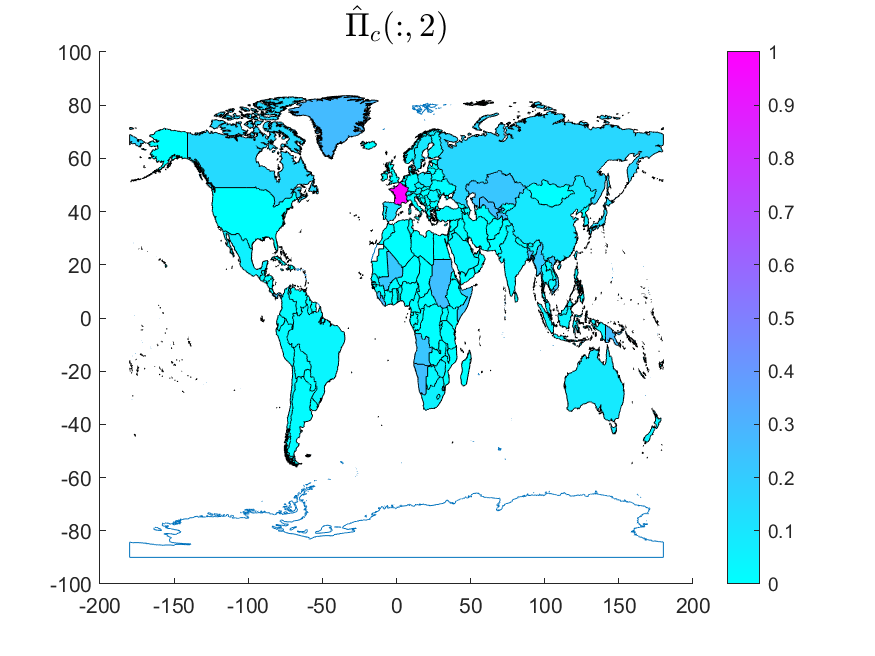}}
{\includegraphics[width=0.495\textwidth]{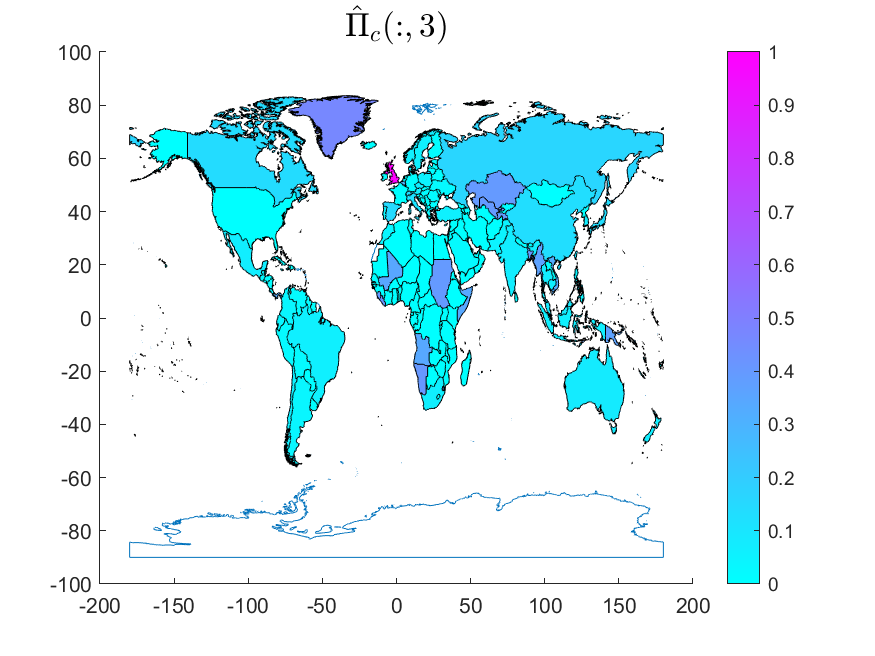}}
{\includegraphics[width=0.495\textwidth]{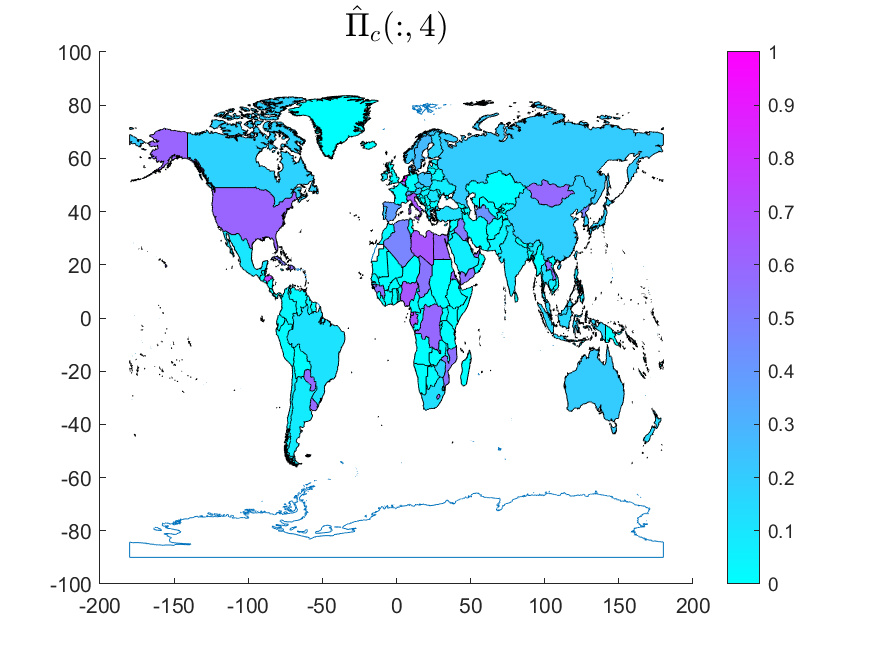}}
{\includegraphics[width=0.495\textwidth]{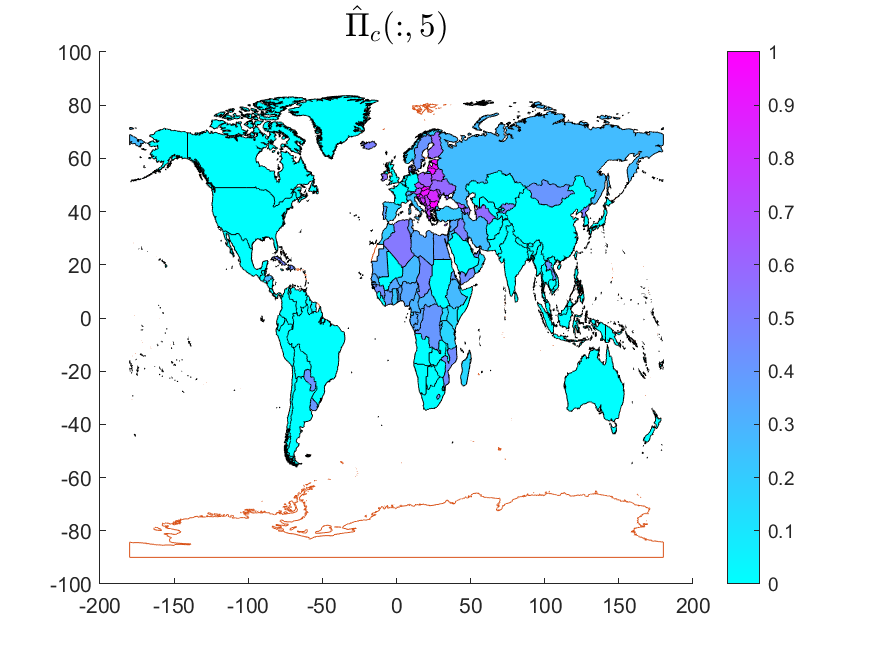}}
{\includegraphics[width=0.495\textwidth]{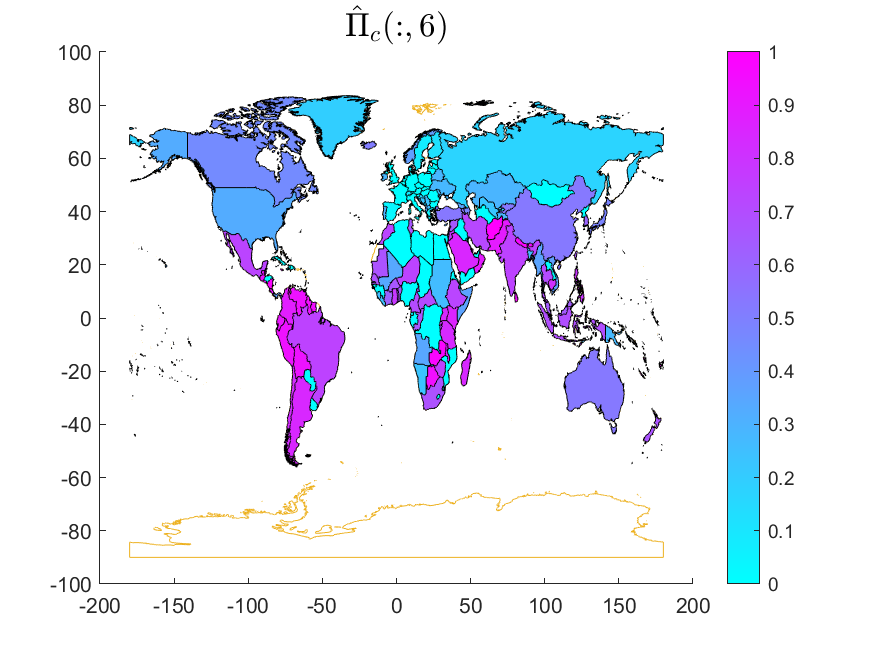}}
\caption{Global map of estimated mixed membership intensity in the six export communities for all the 213 countries.}
\label{FAOK6WorldMapC}
\end{figure}
\section{Conclusion}\label{sec7}
This article contributes to the field of overlapping community detection in multi-layer directed networks by introducing both a novel model and an effective algorithm. The proposed multi-layer MM-ScBM model offers a novel framework that extends existing models to accommodate the complexities of multi-layer directed networks. The spectral procedure we have developed for estimating nodes' memberships enjoys consistent estimation properties under the proposed model. Our theoretical analysis has revealed that factors such as increased sparsity, larger node sets, and more layers in the network can enhance the accuracy of overlapping community detection. Extensive numerical experiments have not only verified our theoretical claims but also highlighted the outstanding performance of our method when compared to its competitors. The application of our algorithm to a real-world multi-layer directed network has produced insightful results, highlighting the potential of our approach in uncovering meaningful overlapping community structures. In summary, our work offers powerful tools for the analysis of multi-layer directed networks.

Future works could focus on several directions to further advance the field of overlapping community detection in multi-layer directed networks. Firstly, extending the proposed multi-layer MM-ScBM model to accommodate dynamic changes in the network structure, such as the evolution of overlapping communities over time, would be a valuable contribution. Secondly, investigating the potential of incorporating additional network features, such as node attributes or edge weights, into the multi-layer MM-ScBM model could enhance its expressiveness and improve the accuracy of overlapping community detection. Thirdly, developing efficient methods to estimate the number of communities in multi-layer networks represents a challenging and promising direction. Lastly, accelerating the proposed spectral procedure would be crucial for handling large-scale multi-layer directed networks.
\section*{CRediT authorship contribution statement}
\textbf{Huan Qing:} Conceptualization, Data curation, Formal analysis, Funding acquisition, Investigation, Methodology, Project administration, Resources, Software, Supervision, Validation, Visualization, Writing – original draft, Writing - review \& editing.
\section*{Declaration of competing interest}
The author declares no competing interests.
\section*{Data availability}
Data will be made available on request. The MATLAB codes for the proposed algorithm are provided in \ref{MatlabCodes}.
\section*{Acknowledgements}
Qing's work was sponsored by the Scientific Research Foundation of Chongqing University of Technology (Grant No. 2024ZDR003), the Science and Technology Research Program of Chongqing Municipal Education Commission (Grant No. KJQN202401168), and the Natural Science Foundation of Chongqing, China (Grant No: CSTB2023NSCQ-LZX0048).
\appendix
\section{Proofs of theoretical results}\label{SecProofs}
\subsection{Proof of Lemma \ref{EigOsum2}.}
\begin{proof} By the pure node condition, $\mathrm{rank}(\Pi_{r})=K$ and $\mathrm{rank}(\Pi_{c})=K$. Since $\tilde{S}_{r}=\Pi_{r}(\rho^{2}\sum_{l\in[L]}B_{l}\Pi'_{c}\Pi_{c}B'_{l})\Pi'_{r}$, we have $\mathrm{rank}(\tilde{S}_{r})=\mathrm{rank}(\sum_{l\in[L]}B_{l}B'_{l})$. Therefore, when $\mathrm{rank}(\sum_{l\in[L]}B_{l}B'_{l})=K$, the rank of $\tilde{S}_{r}$ is $K$, i.e., $\tilde{S}_{r}$ is a low rank matrix with only $K$ nonzero eigenvalues. Let $\tilde{S}_{r}=U_{r}\Lambda_{r}U'_{r}$ be the leading $K$ eigendecomposition of $\tilde{S}_{r}$. Thus, $U_{r}\Lambda_{r}U'_{r}=\Pi_{r}(\rho^{2}\sum_{l\in[L]}B_{l}\Pi'_{c}\Pi_{c}B'_{l})\Pi'_{r}$, which gives that $U_{r}=\Pi_{r}(\rho^{2}\sum_{l\in[L]}B_{l}\Pi'_{c}\Pi_{c}B'_{l})\Pi'_{r}U_{r}\Lambda^{-1}_{r}$. Setting $X_{r}=(\rho^{2}\sum_{l\in[L]}B_{l}\Pi'_{c}\Pi_{c}B'_{l})\Pi'_{r}U_{r}\Lambda^{-1}_{r}$ gives $U_{r}=\Pi_{r}X_{r}$. Since $U_{r}(\mathcal{I}_{r},:)=(\Pi_{r}X_{r})(\mathcal{I}_{r},:)=\Pi_{r}(\mathcal{I}_{r},:)X_{r}=X_{r}$, we have $X_{r}=U_{r}(\mathcal{I}_{r},:)$, i.e., $U_{r}=\Pi_{r}U_{r}(\mathcal{I}_{r},:)$. Similarly, $U_{c}=\Pi_{c}U_{c}(\mathcal{I}_{c},:)$.
\end{proof}
\subsection{Proof of Theorem \ref{MAIN}}
\begin{proof}
First, we have the following lemma.
\begin{lem}\label{boundSumInfinity}
When Assumption \ref{Assum2} holds, with probability $1-o(\frac{1}{n+L})$,
\begin{align*}
\|S_{r}-\tilde{S}_{r}\|_{\infty}=O(\sqrt{\rho^{2}n^{2}L\mathrm{log}(n+L)})+O(\rho^{2}nL), \|S_{c}-\tilde{S}_{c}\|_{\infty}=O(\sqrt{\rho^{2}n^{2}L\mathrm{log}(n+L)})+O(\rho^{2}nL).
\end{align*}
\end{lem}
\begin{proof}
For $\|S_{r}-\tilde{S}_{r}\|_{\infty}$, we have
\begin{align*}
 \|S_{r}-\tilde{S}_{r}\|_{\infty}&=\mathrm{max}_{i\in[n]}(\sum_{j\neq i,j\in[n]}|\sum_{l\in[L]}\sum_{m\in[n]}(A_{l}(i,m)A_{l}(j,m)-\Omega_{l}(i,m)\Omega_{l}(j,m))|+\sum_{l\in[L]}\sum_{m\in[n]}\Omega^{2}_{l}(i,m))\\
 &\leq \mathrm{max}_{i\in[n]}(\sum_{j\neq i,j\in[n]}|\sum_{l\in[L]}\sum_{m\in[n]}(A_{l}(i,m)A_{l}(j,m)-\Omega_{l}(i,m)\Omega_{l}(j,m))|+\sum_{l\in[L]}\sum_{m\in[n]}\rho^{2})\\
 &=\mathrm{max}_{i\in[n]}(\sum_{j\neq i,j\in[n]}|\sum_{l\in[L]}\sum_{m\in[n]}(A_{l}(i,m)A_{l}(j,m)-\Omega_{l}(i,m)\Omega_{l}(j,m))|)+\rho^{2}nL.
\end{align*}
Let $x$ be any vector in $\mathbb{R}^{(n-1)\times1}$, $y_{(ij)}=\sum_{l\in[L]}\sum_{m\in[n]}(A_{l}(i,m)A_{l}(j,m)-\Omega_{l}(i,m)\Omega_{l}(j,m))$ for $i\in[n],j\neq i,j\in[n]$, and $T_{(i)}=\sum_{j\neq i,j\in[n]}y_{(ij)}x(j)$ for $i\in[n]$. Set $\tau=\mathrm{max}(\mathrm{max}_{i\in[n]}\mathrm{max}_{j\in[n]}|\sum_{l\in[L]}\sum_{m\in[n]}(A_{l}(i,m)A_{l}(j,m)-\Omega_{l}(i,m)\Omega_{l}(j,m))|,\\
\mathrm{max}_{i\in[n]}\mathrm{max}_{j\in[n]}|\sum_{l\in[L]}\sum_{m\in[n]}(A_{l}(m,i)A_{l}(m,j)-\Omega_{l}(m,i)\Omega_{l}(m,j))|)$.  For $i\in[n],j\neq i,j\in[n]$, we have $\mathbb{E}(y_{(ij)}x(j))=0$ since $j\neq i$, $|y_{(ij)}x(j)|\leq\tau\|x\|_{\infty}$, and
\begin{align*}
\sum_{j\neq i,j\in[n]}\mathbb{E}[y^{2}_{(ij)}x^{2}(j)]&=\sum_{j\neq i,j\in[n]}x^{2}(j)\mathbb{E}[y^{2}_{(ij)}]=\sum_{j\neq i,j\in[n]}x^{2}(j)\sum_{l\in[L]}\sum_{m\in[n]}\Omega_{l}(i,m)\Omega_{l}(j,m)(1-\Omega_{l}(i,m)\Omega_{l}(j,m))\\
&\leq\sum_{j\neq i,j\in[n]}x^{2}(j)\sum_{l\in[L]}\sum_{m\in[n]}\Omega_{l}(i,m)\Omega_{l}(j,m)\leq\sum_{j\neq i,j\in[n]}x^{2}(j)\sum_{l\in[L]}\sum_{m\in[n]}\rho^{2}=\rho^{2}\|x\|^{2}_{F}nL.
\end{align*}
By Theorem 1.4 (Bernstein inequality) of \cite{tropp2012user}, let $t$ be any nonnegative value, we have
\begin{align*}
\mathbb{P}(|T_{(i)}|\geq t)\leq\mathrm{exp}(\frac{-t^{2}/2}{\rho^{2}\|x\|^{2}_{F}nL+\frac{\tau\|x\|_{\infty}t}{3}}).
\end{align*}
Let $t$ be $\sqrt{\rho^{2}\|x\|^{2}_{F}nL\mathrm{log}(n+L)}\times\frac{\alpha+1+\sqrt{(\alpha+1)(\alpha+19)}}{3}$ for any $\alpha\geq0$. When $\rho^{2}\|x\|^{2}_{F}nL\geq\tau^{2}\|x\|^{2}_{\infty}\mathrm{log}(n+L)$ holds, the following inequality holds:
\begin{align*}
\mathbb{P}(|T_{(i)}|\geq t)\leq\mathrm{exp}(-(\alpha+1)\mathrm{log}(n+L)\frac{1}{\frac{18}{(\sqrt{\alpha+1}+\sqrt{\alpha+19})^{2}}+\frac{2\sqrt{\alpha+1}}{\sqrt{\alpha+1}+\sqrt{\alpha+19}}\sqrt{\frac{\tau^{2}\|x\|^{2}_{\infty}\mathrm{log}(n+L)}{\rho^{2}\|x\|^{2}_{F}nL}}})\leq\frac{1}{(n+L)^{\alpha+1}}.
\end{align*}
Since $x$ can be any $(n-1)\times1$ vector and $\alpha$ is any nonnegative value, letting $x$'s elements be in $\{-1,1\}$ and $\alpha=1$ gives: when $\rho^{2}n^{2}L\geq\tau^{2}\mathrm{log}(n+L)$, with probability $1-o(\frac{1}{n+L})$,
\begin{align*}
\mathrm{max}_{i\in[n]}T_{(i)}=O(\sqrt{\rho^{2}n^{2}L\mathrm{log}(n+L)}),
\end{align*}
which gives $\|S_{r}-\tilde{S}_{r}\|_{\infty}=O(\sqrt{\rho^{2}n^{2}L\mathrm{log}(n+L)})+O(\rho^{2}nL)$.
Similarly, $\|S_{c}-\tilde{S}_{c}\|_{\infty}=O(\sqrt{\rho^{2}n^{2}L\mathrm{log}(n+L)})+O(\rho^{2}nL)$. For simplicity, we restrict our consideration to the highly sparse regime $\tau\leq c_{3}$ for some constant $c_{3}>0$. For this sparse regime, the requirement $\rho^{2}n^{2}L\geq\tau^{2}\mathrm{log}(n+L)$ reduces to Assumption \ref{Assum2}.
\end{proof}

According to Theorem 4.2 \citep{cape2019the}, when $|\lambda_{K}(\tilde{S}_{r})|\geq 4\|S_{r}-\tilde{S}_{r}\|_{\infty}$ and $|\lambda_{K}(\tilde{S}_{c})|\geq 4\|S_{c}-\tilde{S}_{c}\|_{\infty}$ hold, there are two orthogonal matrices $\mathcal{O}_{r}$ and $\mathcal{O}_{c}$ such that
\begin{align*}
\|\hat{U}_{r}-U_{r}\mathcal{O}_{r}\|_{2\rightarrow\infty}\leq14\frac{\|S_{r}-\tilde{S}_{r}\|_{\infty}\|U_{r}\|_{2\rightarrow\infty}}{|\lambda_{K}(\tilde{S}_{r})|},
\|\hat{U}_{c}-U_{c}\mathcal{O}_{c}\|_{2\rightarrow\infty}\leq14\frac{\|S_{c}-\tilde{S}_{c}\|_{\infty}\|U_{c}\|_{2\rightarrow\infty}}{|\lambda_{K}(\tilde{S}_{c})|}.
\end{align*}
Because $\varpi_{r}:=\|\hat{U}_{r}\hat{U}'_{r}-U_{r}U'_{r}\|_{2\rightarrow\infty}\leq2\|\hat{U}_{r}-U_{r}\mathcal{O}_{r}\|_{2\rightarrow\infty}$ and $\varpi_{c}:=\|\hat{U}_{c}\hat{U}'_{c}-U_{c}U'_{c}\|_{2\rightarrow\infty}\leq2\|\hat{U}_{c}-U_{c}\mathcal{O}_{c}\|_{2\rightarrow\infty}$, we have
\begin{align*}
\varpi_{r}\leq28\frac{\|S_{r}-\tilde{S}_{r}\|_{\infty}\|U_{r}\|_{2\rightarrow\infty}}{|\lambda_{K}(\tilde{S}_{r})|}, \varpi_{c}\leq28\frac{\|S_{c}-\tilde{S}_{c}\|_{\infty}\|U_{c}\|_{2\rightarrow\infty}}{|\lambda_{K}(\tilde{S}_{c})|}.
\end{align*}
Based on Lemma 3.1 \citep{mao2021estimating} and Condition \ref{condition}, we have $\|U_{r}\|_{2\rightarrow\infty}=O(\sqrt{\frac{1}{n}})$ and $\|U_{c}\|_{2\rightarrow\infty}=O(\sqrt{\frac{1}{n}})$, which give that
\begin{align*}
\varpi_{r}=O(\frac{\|S_{r}-\tilde{S}_{r}\|_{\infty}}{|\lambda_{K}(\tilde{S}_{r})|\sqrt{n}}), \varpi_{c}=O(\frac{\|S_{c}-\tilde{S}_{c}\|_{\infty}}{|\lambda_{K}(\tilde{S}_{c})|\sqrt{n}}).
\end{align*}
Combining Assumption \ref{Assum22} and Condition \ref{condition} gives the following result:
\begin{align*}
|\lambda_{K}(\tilde{S}_{r})|&=\sqrt{\lambda_{K}((\sum_{l\in[L]}\Omega_{l}\Omega'_{l})^{2})}=\sqrt{\lambda_{K}((\sum_{l\in[L]}\rho^{2}\Pi_{r} B_{l}\Pi'_{c}\Pi_{c} B'_{l}\Pi'_{r})^{2})}=\rho^{2}\sqrt{\lambda^{2}_{K}(\sum_{l\in[L]}\Pi_{r} B_{l}\Pi'_{c}\Pi_{c}B'_{l}\Pi'_{r})}\\
&=\rho^{2}\sqrt{\lambda^{2}_{K}(\Pi_{r}(\sum_{l\in[L]}B_{l}\Pi'_{c}\Pi_{c} B'_{l})\Pi'_{r})}=\rho^{2}\sqrt{\lambda^{2}_{K}(\Pi'_{r}\Pi_{r}(\sum_{l\in[L]}B_{l}\Pi'_{c}\Pi_{c} B'_{l}))}\geq\rho^{2}\lambda_{K}(\Pi'_{r}\Pi_{r})\sqrt{\lambda^{2}_{K}(\sum_{l\in[L]}B_{l}\Pi'_{c}\Pi_{c} B'_{l})}\\
&=O(\rho^{2}\lambda_{K}(\Pi'_{r}\Pi_{r})\lambda_{K}(\Pi'_{c}\Pi_{c})|\lambda_{K}(\sum_{l\in[L]}B_{l}B'_{l})|)=O(\rho^{2}n^{2}L).
\end{align*}
Similarly, we have $|\lambda_{K}(\tilde{S}_{c})|\geq O(\rho^{2}n^{2}L)$. Therefore, we have
\begin{align*}
\varpi_{r}=O(\frac{\|S_{r}-\tilde{S}_{r}\|_{\infty}}{\rho^{2}n^{2.5}L}), \varpi_{c}=O(\frac{\|S_{c}-\tilde{S}_{c}\|_{\infty}}{\rho^{2}n^{2.5}L}).
\end{align*}
By the proof of Theorem 3.2 of \citep{mao2021estimating}, we have
\begin{align*}
&\mathrm{max}_{i\in[n]}\|e'_{i}(\hat{\Pi}_{r}-\Pi_{r}\mathcal{P}_{r})\|_{1}=O(\varpi_{r}\kappa(\Pi'_{r}\Pi_{r})\sqrt{\lambda_{1}(\Pi'_{r}\Pi_{r})})=O(\varpi_{r}\sqrt{\frac{n}{K}})=O(\varpi_{r}\sqrt{n})=O(\frac{\|S_{r}-\tilde{S}_{r}\|_{\infty}}{\rho^{2}n^{2}L}),\\
&\mathrm{max}_{i\in[n]}\|e'_{i}(\hat{\Pi}_{c}-\Pi_{c}\mathcal{P}_{c})\|_{1}=O(\varpi_{c}\kappa(\Pi'_{c}\Pi_{c})\sqrt{\lambda_{1}(\Pi'_{c}\Pi_{c})})=O(\varpi_{c}\sqrt{\frac{n}{K}})=O(\varpi_{c}\sqrt{n})=O(\frac{\|S_{c}-\tilde{S}_{c}\|_{\infty}}{\rho^{2}n^{2}L}),
\end{align*}
where $\mathcal{P}_{r}$ and $\mathcal{P}_{c}$ are two $K\times K$ permutation matrices. By Lemma \ref{boundSumInfinity}, this theorem holds.
\end{proof}
\section{MATLAB codes of CSPDSoS}\label{MatlabCodes}
The MATLAB codes of CSPDSoS are provided below:
\begin{lstlisting}
function [Pi_r, Pi_c] = CSPDSoS(A_all, K)
    % Inputs:
    %   A_all - A tensor with A_all(:,:,l) being the adjacency matrix of layer l
    %   K - Number of communities to detect
    % Outputs:
    %   Pi_r - Estimated row membership matrix
    %   Pi_c - Estimated column membership matrix

    % Number of nodes
    n = size(A_all, 1);
    % Number of layers
    L = size(A_all, 3);

    % Initialize the debiased aggregation matrices
    S_r = zeros(n, n);
    S_c = zeros(n, n);

    % Compute the debiased aggregation matrices
    for l = 1:L
        Al = A_all(:, :, l);
        Drl = diag(sum(Al, 2)); % Diagonal matrix of row sums
        Dcl = diag(sum(Al, 1)); % Diagonal matrix of column sums
        S_r = S_r + Al * Al' - Drl; % Debiased aggregation matrix for sending patterns
        S_c = S_c + Al' * Al - Dcl; % Debiased aggregation matrix for receiving patterns
    end

    % Compute the leading K eigenvectors of S_r and S_c
    [U_r, ~] = eigs(S_r, K); % Leading K eigenvectors for sending patterns
    [U_c, ~] = eigs(S_c, K); % Leading K eigenvectors for receiving patterns

    % Find the indices of the K purest nodes using SPA
    [I_r] = SPA(U_r, K); % Indices of pure nodes for sending patterns
    [I_c] = SPA(U_c, K); % Indices of pure nodes for receiving patterns

    % Compute the membership matrix for sending patterns
    C_r = U_r(I_r, :); % Corner matrix in U_r
    Y_r = max(0, U_r / C_r);
    Pi_r = zeros(n, K); % Initialize row-wise membership matrix
    for i = 1:n
        Pi_r(i, :) = Y_r(i, :) / sum(Y_r(i, :));
    end

    % Compute the membership matrix for receiving patterns
    C_c = U_c(I_c, :); % Corner matrix in U_c
    Y_c = max(0, U_c / C_c);
    Pi_c = zeros(n, K); % Initialize column-wise membership matrix
    for i = 1:n
        Pi_c(i, :) = Y_c(i, :) / sum(Y_c(i, :));
    end

    % Nested function: SPA (Successive Projection Algorithm)
    function [pure] = SPA(X, K)
        % SPA: Successive Projection Algorithm
        % Inputs:
        %   X - Feature matrix
        %   K - Number of purest nodes to select
        % Outputs:
        %   pure - Indices of the K purest nodes

        % Initialize empty array for pure node indices
        pure = [];
        % Compute row norms of the feature matrix
        row_norm = vecnorm(X').^2';

        % Iterate K times to select the K purest nodes
        for i = 1:K
            % Find the node with the maximum row norm
            [~, idx_tmp] = max(row_norm);
            % Add the index to the list of pure nodes
            pure = [pure idx_tmp];
            % Store the eigenvector of the selected pure node
            U(i, :) = X(idx_tmp, :);
            for j = 1:i-1
                U(i, :) = U(i, :) - U(i, :) * (U(j, :)' * U(j, :));
            end
            % Normalize the orthogonalized eigenvector
            U(i, :) = U(i, :) / norm(U(i, :));
            row_norm = row_norm - (X * U(i, :)').^2;
        end
    end
end
\end{lstlisting}
\bibliographystyle{elsarticle-num}
\bibliography{refMLBiMMSB}
\end{document}